\documentclass[letterpaper, 10 pt, conference]{ieeeconf}      

\IEEEoverridecommandlockouts                              

\usepackage[T1]{fontenc}   

\newcommand{\markup}[1]{{\color{black}{#1}}}

\usepackage{geometry}
\usepackage{amsthm}      
\usepackage{amsmath,amsfonts}
\usepackage{bbm}
\usepackage{amssymb}  
\usepackage{algorithmic}
\usepackage{algorithm}
\usepackage{array} 
\usepackage{textcomp}
\usepackage{stfloats}
\usepackage{url}
\usepackage{verbatim}
\usepackage{graphicx}
\usepackage{subcaption}
\usepackage{tabularx}
\usepackage{siunitx}
\DeclareMathOperator*{\argmin}{argmin} 

\usepackage{bondgraphs}
\usepackage{tikz}
\usetikzlibrary{calc}
\usepackage{epstopdf}
\usetikzlibrary{shapes,fit}

\newtheorem{theorem}{Theorem}
\newtheorem{example}{Example}

\newtheorem{proposition}{Proposition}
\newtheorem*{objective*}{PBC objective}
\newtheorem*{tankobjective*}{Sampled PBC objective}
\newtheorem{remark}{Remark}

\usepackage{color}

\usepackage[style=ieee,
			url=false,
			isbn=false,
			doi=true,
			giveninits=true,
			uniquename=init,
			isbn=false,
			backend=bibtex,
			minbibnames=1,
			maxbibnames=6	]{biblatex}

\bibliography{root1.bib}

\title{\LARGE \bf
The effect of control barrier functions on energy transfers in controlled physical systems
}

\author{Federico Califano$^{1}$, Riccardo Zanella$^{1}$, Alessandro Macchelli$^{2}$, Stefano Stramigioli$^{1}$
\thanks{$^{1}$Robotics and Mechatronics (RaM) group, University of Twente, The Netherlands. email:
        \{f.califano, r.zanella, 
       s.stramigioli\}@utwente.nl.}
\thanks{$^{2}$Department of Electrical, Electronic and Information Engineering (DEI), University of Bologna, Italy. email: a.macchelli@unibo.it.}
}

\begin{document}

\maketitle
\thispagestyle{empty}
\pagestyle{empty}

\begin{abstract}
Using a port-Hamiltonian formalism, we show the effects of safety-critical control implemented with control barrier functions (CBFs) on the power balance of controlled physical systems. The presented results will provide novel tools to design CBFs inducing desired energetic behaviors of the closed-loop system, including non-trivial damping injection effects and non-passive control actions, effectively injecting energy \markup{into} the system in a controlled manner. Simulations validate the \markup{presented} results.
\end{abstract}

\section{Introduction}

Control objectives in engineering applications are becoming increasingly complex in terms of their semantic and mathematical description. An instructive example to think about is safe collaboration between humans and robots \cite{Robla-Gomez2017WorkingEnvironments,Hamad2023AInteraction,Siciliano2016SpringerRobotics}, where control goals are often represented by multi-objective functions representing the specific task to be executed. These objectives (\markup{often} unknown \markup{or} implicitly designed through e.g. control by human demonstrations) are non-trivially intertwined with safety specifications, aimed at representing critical hazards whose prevention needs to be certified before any \markup{attempt at} real-world application. These safety specifications are often conflicting with the main task to be executed, and the semantic distinction between the ``main task", represented by objective functions, and the ``safety specifications", represented by constraints on the state of the system in the optimization problem, is ubiquitous in the control-theoretic literature \cite{Ames2019ControlApplications,Wabersich2023Data-DrivenSystems,Ferraguti2022SafetyApproach}. This view of safety-critical control applications has led to the concept of \textit{safety filters} \cite{Wabersich2023Data-DrivenSystems}, which are computational units with the goal of transforming a nominal control input (designed for the main task, disregarding the safety requirements) into a new control input which i) minimally modifies the nominal controller and ii) ensures the safety requirements. 
One of the most popular techniques for implementing these safety filters is represented by \textit{control barrier functions} (CBF) \cite{Ames2019ControlApplications,Ames2017ControlSystems}, an algorithm which aims to modify a nominal control input to achieve forward invariance of a safe set.

In this paper, we perform an analysis of the effect of CBFs on energy transfers in physical systems, and present novel control tools aimed at regulating the energy in the closed-loop system.
The idea that an energetic \markup{characterization} captures \markup{high-level} behavioral properties of a system is present in the so-called energy-based or \textit{energy-aware} control literature \cite{Duindam2009ModelingSystems,Stramigioli2015Energy-AwareRobotics,Califano2022OnSystems,Ortega2001PuttingControl}: the ability to control the energy transfers in controlled system induces \textit{desired behaviors}, beyond mere stabilization purposes. \markup{Energy-aware designs serve to achieve high level control objectives based on energy, including and extending passivity-based control (PBC) strategies \cite{vanderSchaftL2}, which are used to robustly achieve stabilization goals.} Energy-aware control is particularly interesting in collaborative robotic applications that take place in unstructured environments, possibly in the presence of humans, since energy transfers between the robot and its surroundings encode information on performance, safety, and energy efficiency \cite{Tadele2014CombiningRobots,Laffranchi2009SafeControl,Zacharaki2020SafetySurvey,Califano2022OnSystems,Tadele2014ThePublications}. \markup{Another instructive example embodying the need of an energy-aware design is periodic locomotion tasks, as purely passive designs cannot sustain them without continuous energy injection to counteract internal dissipation.} The narrative of ``putting energy back in control" \cite{Ortega2001PuttingControl} has been the driving motivation in the development of control schemes in the port-Hamiltonian (pH) formalism \cite{Duindam2009ModelingSystems,Rashad2020TwentyReview}. \markup{PH systems encodes energetic properties of autonomous and controlled physical systems, and will be conveniently used in this paper to gain intuitive and technical advantages in the energetic analysis.}
\markup{While energy-aware methods claim to manage energy flow, practical tools to realize this are still missing. This work provides a novel conceptual and methodological approach to energy-aware control, where CBFs are designed to actively regulate energy flows, including controlled energy injections into the closed-loop system.}

\markup{We recognize other works combining the CBF methodology with energetic analysis of the closed-loop system. In \cite{ROMDLONY201639} the authors combine pH modeling and CBFs to obtain a stable closed-loop system with safety guarantees. In \cite{Capelli2022PassivityEnergy} CBFs are used in combination with \textit{energy tanks} \cite{Califano2022OnSystems} to passivize a potentially non passive control action. In \cite{Notomista2019Passivity-BasedFunctions} a time varying CBF is designed to passivize a closed-loop system. \cite{10787049} develops an optimization that makes the closed-loop system passive and safe at the same time. In \cite{10136379} a class of CBFs preserving the passivity of a closed-loop system were characterized with technical tools similar to those used in this paper. \cite{10886724} implements a similar idea for switching systems.
Other relevant works are \cite{Singletary2021Safety-CriticalSystems}, which introduces so called \textit{energy-based} CBFs to deal with kinematic constraints in robotics, and \cite{9561981}, where a CBF-inspired controller is developed in a similar context and comprising a passivity and safety analysis.
In all these works the desired energetic properties of the closed-loop system were limited to passivity, considered as a necessary condition for some form of safety. In this work we will extend the use of CBFs to energy-aware designs.
}
\markup{In particular our contributions include the introduction of CBFs along}: i) damping injection schemes, in which energy is extracted from the controlled system, \markup{with novel sufficient conditions on closed-loop stability}; and ii) new schemes which effectively inject energy in the controlled system, going beyond passive designs and giving new perspectives to achieve desired closed-loop energetic behaviors. \markup{This last contribution is validated through a task resembling periodic locomotion, introducing a completely new perspective on the use of CBFs}.

The paper is organized as follows. In Sec. \ref{sec:background} the background on pH systems and CBFs is reviewed. In Sec. \ref{sec:3} the main contribution of this letter is given and applied to mechanical systems in Sec. \ref{sec:4}. In Sec. \ref{sec:sims} we show numerical simulations and Sec. \ref{sec:conc} contains conclusions and future work.

\section{Background}
\label{sec:background}
We refer to \cite{vanderSchaftL2,Duindam2009ModelingSystems} regarding pH systems and to \cite{Ames2019ControlApplications,Xu2015RobustnessControl,Ames2017ControlSystems,Singletary2021Safety-CriticalSystems} regarding CBF for references that fully cover the presented background.

\subsection{Port-Hamiltonian systems}
The input--state--output representation of a port--Hamiltonian (pH) system is: 
\begin{equation}
\label{eq:pH}
    \begin{cases}
      \dot{x}= (J(x)-R(x))\partial_x H(x)+g(x)u \\
      y=g(x)^{\top}\partial_x H(x)
    \end{cases} 
\end{equation}
where $x\in \mathcal{D}\subseteq \mathbb R^n$ is the state, $u\in \mathcal{U}\subseteq \mathbb R^m$ is the input, $g(x)$ is the input matrix, $J(x)=-J(x)^{\top}$ and $R(x)=R(x)^{\top}\geq 0$ are respectively skew-symmetric and positive semi-definite symmetric matrices representing the power-preserving and the dissipative components of the system. The non-negative function $H: \mathcal{D}\to \mathbb{R}^{+}$ is called the \textit{Hamiltonian} and maps the state into the total physical energy of the system. $\partial_x H(x) \in \mathbb R^{n}$ denotes the gradient of $H$, represented as a column, and $\partial_x^{\top}H(x)$ denotes its transposed. When not explicitly stated, all described variables are assumed to have a degree of continuity such that the right-hand side of (\ref{eq:pH}) is locally Lipschitz, to guarantee the existence and uniqueness of the solutions.
We now present geometric properties of port-Hamiltonian systems that are relevant in this paper. 

\subsubsection{Power-preserving structure}
The skew-symmetric matrix operator $J(x)$ induces the bracket $\{ \cdot,\cdot \}_J$, a skew-symmetric bilinear map that takes as input two smooth scalar functions on the state space to produce another scalar function. Given two scalar functions $A: \mathcal{D}\to \mathbb{R}$ and $B: \mathcal{D}\to \mathbb{R}$, the bracket is defined as $\{ A,B\}_J := \partial_x^{\top}A(x) J(x) \partial_x B(x)$. When the second slot of the bracket is fed with the Hamiltonian $H$, the bracket completely represents the conservative Hamiltonian dynamics obtained from (\ref{eq:pH}) setting $R(x)=0$ and $u=0$. In this case, for any function $A: \mathcal{D}\to \mathbb{R}$, its variation along the solution of (\ref{eq:pH}) is $\dot{A} = \{A, H \}_J$, and as a particular case the conservation of energy is encoded in the skew-symmetry of the bracket since $\dot{H}=\{H,H\}_J=0$.
\subsubsection{Dissipation structure}
We denote by the bracket $[ \cdot,\cdot ]_Y$ the bilinear map $[ A,B]_Y := \partial_x^{\top}A(x) Y(x) \partial_x B(x)$ for a state-dependent symmetric matrix $Y(x) \in \mathbb{R}^{n \times n}$. For system (\ref{eq:pH}), when $Y(x)=R(x)$ and the bracket is fed twice with the Hamiltonian, this bracket represents dissipated power due to the dissipative effects modeled in $R(x)$, i.e., $[H,H ]_R=\partial_x ^{\top} H(x) R(x) \partial_x H(x)\geq 0$.
\subsubsection{Passivity}
The input $u$ and the output $y$ in (\ref{eq:pH}) are \textit{co-located}, in the sense that the variation of energy due to the input is $\partial_x^{\top} H(x) g(x) u=y^{\top} u$.
It follows that system (\ref{eq:pH}) is \textit{passive} with the Hamiltonian $H(x)$ as storage function and input-output pair $(u,y)$:
\begin{equation}
    \label{eq:passivity}\dot{H}=\partial_x^{\top}H(x) \dot{x}=\underbrace{\{H,H \}_J}_{0} -\underbrace{[H,H ]_R}_{\geq 0} + y^{\top}u\leq y^{\top}u.
\end{equation}
The latter power balance is a statement of energy conservation for the physical system (\ref{eq:pH}), where \markup{all ways in which physical energy can flow along the system are displayed}: the skew-symmetric structure $J$ represents pure routing of energy, the dissipative structure $R$ dissipated energy, and the duality product $y^{\top}u$ represents the instantaneous power injected by the input. Passivity condition (\ref{eq:passivity}) states that the variation of energy in the system is bounded by the injected power.
\subsubsection{Control}\label{subsub:control}From a system-theoretic perspective passivity implies stability under weak conditions: $H(x)$ is a valid Lyapunov candidate since $u=0$ implies $\dot{H}\leq 0$. Since (\ref{eq:pH}) is time-invariant, the autonomous system converges to the largest invariant subset of $\{x \in \mathcal{D} | [H,H]_R=0\}$ which generally depends on the dissipation matrix $R(x)$ and the Hamiltonian $H(x)$.
Many pH-inspired control schemes have been proposed, with the rationale of obtaining a closed-loop system in the form (\ref{eq:pH}) with Hamiltonian and system matrices corresponding to desired behaviors. Importantly, IDA-PBC designs \cite{ORTEGA2002585}, indeed aim at obtaining a closed-loop port-Hamiltonian system \markup{encoding} desired energetic properties.

\subsection{Control barrier functions}

Control barrier functions (CBFs) represent a technique to guarantee forward invariance of a set $\mathcal{C}$, normally called \textit{safe set}. We present the standard background for CBFs directly specialised to port-Hamiltonian systems (\ref{eq:pH}), providing interpretability of CBF-based algorithms from an energetic perspective. 

The goal is to design a state feedback $u=k(x)$ resulting in the closed-loop system $\dot{x}=(J(x)-R(x))\partial_x H+g(x)k(x)$ such that

\begin{equation}
    \forall x(0) \in \mathcal{C} \implies x(t)\in \mathcal{C} \,\,\, \forall t>0.
\end{equation}
The safe set $\mathcal{C}$ is built as the superlevel set of a continuously differentiable function $h:\mathcal{D}\to \mathbb{R}$, i.e., $
    \mathcal{C} = \{ x\in \mathcal{D} : h(x)\geq0 \}.$
The function $h(x)$ is then defined as a CBF on $\mathcal{D}$ if $\partial_x h(x)\neq 0, \forall x \in \partial \mathcal{C} = \{ x\in \mathcal{D} : h(x) = 0 \}$, and 

\begin{equation}
 \sup_{u \in \mathcal{U}} [\dot{h}(x,u)] \geq -\alpha(h(x))   
\end{equation}
for all $x\in \mathcal{D}$ and some \textit{extended class $\mathcal{K}$ function}\footnote{A function $\alpha: (-b,a) \to (- \infty, \infty)$ with $a,b>0$, which is continuous, strictly increasing, and $\alpha(0)=0$.} $\alpha$. Here $\dot{h}=\partial_x^{\top}h(x) \dot{x}$ \markup{is used} to denote the time derivative of $h$ along the solution of (\ref{eq:pH}), which results in:
\begin{equation}
\label{eq:closedFormPsi}
    \dot{h}(x,u)=\{h,H \}_J -[h,H ]_R + \partial_x^{\top} h(x)g(x)u.
\end{equation}
The link between the existence of a CBF and the forward invariance of the related safe set is established by the following key result, present in \cite{Ames2017ControlSystems} and conveniently reported here applied to pH systems (\ref{eq:pH}).

\begin{theorem}
\label{thm:1}
Let $h(x)$ be a CBF on $\mathcal{D}$ for (\ref{eq:pH}). Any locally Lipschitz controller $u=k(x)$ such that $\{h,H \}_J -[h,H ]_R + \partial_x ^{\top}h (x)g(x)k(x) \geq -\alpha (h(x))$ provides forward invariance of $\mathcal{C}$. Additionally $\mathcal{C}$ is asymptotically stable on $\mathcal{D}$.
\end{theorem}
The way controller synthesis induced by CBFs are implemented is to use them as \textit{safety filters}, transforming a nominal state-feedback control input $u_{\textrm{nom}}(x)$ into a new state-feedback control input $u^*(x)$ in a minimally invasive fashion in order to guarantee forward invariance of $\mathcal{C}$. In practice, the following Quadratic Program (QP) is solved:

\begin{equation}
\label{eq:LQ}
\begin{aligned}
u^*(x)= & \argmin_{u\in \mathcal{U}} \quad  ||u-u_{\textrm{nom}}(x) ||^2\\
 & \textrm{s.t.}  \quad  \dot{h}(x,u) \geq -\alpha(h(x)) 
\end{aligned}
\end{equation}

The transformation of the desired control input $u_{\textrm{nom}}(x)$ in $u^*(x)$ by solving (\ref{eq:LQ}) is denoted as \textit{safety-critical control}.
A result that will be crucially used in this work is the following, presented in \cite{Xu2015RobustnessControl} and presented here on pH systems (\ref{eq:pH}) \markup{with the notation used in this paper}.
\markup{
\begin{theorem}   
\label{thm:closed-form}
Let $h(x)$ be a CBF on $\mathcal{D}$ for (\ref{eq:pH}) and assume $\mathcal{U}=\mathbb{R}^{m}$. Define $\Psi(x):=\dot{h}(x,u_{\textup{nom}}(x))+\alpha(h(x))$ and assume $[h,h]_{g g^{\top}}\neq 0$ when $\Psi<0$. The Lipschitz continuous closed-form solution for (\ref{eq:LQ}) is given by $u^*(x)=u_{\textup{nom}}(x)+u_{\textup{safe}}(x)$, where 
\begin{equation}
\label{eq:safetyComponent}
u_{\textup{safe}}(x)= 
 -\mathbb{I}_{\Psi<0}\frac{g^{\top}\partial_x h}{ [h,h]_{g g^{\top}}}\Psi,
\end{equation}
where the indicator function $\mathbb{I}_{\Psi<0}$ returns $1$ if $\Psi < 0$ ($u_{\textrm{safe}}\neq0$) and $0$ otherwise  ($u_{\textrm{safe}}=0$).
\end{theorem}}
\markup{Consistently with the arguments reported in Sec. \ref{sec:background}.A.4, in the remainder of the paper, (\ref{eq:pH}) will represent a (possibly controlled) physical system, encoding desired closed-loop properties. This perspective, which is custom in energy-based frameworks (e.g., in IDA-PBC designs \cite{ORTEGA2002585}), introduces a new perspective in the CBF framework, i.e., the desired closed-loop behavior is achieved with $u_\textrm{nom}=0$, which will be assumed throughout the paper.
}          
\section{Putting energy back in control through CBFs}
\label{sec:3}

In this section, we discuss how CBF-induced safety-critical control influences the power balance of the controlled physical system. \markup{The following result presents a closed-form solution for the power injected by the safety-critical controller (\ref{eq:safetyComponent})}.

\begin{theorem}
\markup{Under the assumptions of Theorem \ref{thm:closed-form}}, the power balance for the pH system (\ref{eq:pH}), undergoing safety-critical control (\ref{eq:LQ}) induced by a CBF $h(x)$, is

\begin{equation}
\label{eq: dissipationDiscussion}
    \dot{H}=-[H,H]_{R} -\mathbb{I}_{\Psi<0} \frac{ [H,h ]_{g g^{\top}}}{ [h,h ]_{g g^{\top}}} \Psi,
\end{equation}
where $\Psi=\{h,H\}_J -[h,H]_R+\alpha(h(x))$.
\end{theorem}
\begin{proof}
 The results follows from an explicit calculation of the power balance (\ref{eq:passivity}) using $u=u_{\textup{safe}}(x)$ given by (\ref{eq:safetyComponent}). \markup{The expression of $\Psi$ follows from (\ref{eq:closedFormPsi}) with $u=u_{\textrm{nom}}=0$.}
\end{proof}
We indicate the instantaneous power injection induced by the safety-critical controller (the last addend in (\ref{eq: dissipationDiscussion})) with
\begin{equation}
\label{eq:safety-criticalPower}
    P_{h,\alpha}=-\mathbb{I}_{\Psi<0} \frac{ [H,h ]_{g g^{\top}}}{ [h,h ]_{g g^{\top}}} \Psi
\end{equation}
where the notation stresses the fact that this term depends both on the CBF $h$ and on the function $\alpha$, a parameter of the CBF algorithm. \markup{We will now explore different control strategies induced by the analytic inspection of this power term.}

\subsubsection{CBF-based damping injection}
\markup{
As reported in Sec. \ref{subsub:control}, when the Hamiltonian $H$ qualifies as a Lyapunov function, the minima of $H$ form a stable set.
The following result gives a useful sufficient condition for preservation of local stability, as well as a constructive method for novel damping injection schemes.

\begin{proposition}
\label{prop:prop}
Suppose $H$ is a Lyapunov function for the pH system (\ref{eq:pH}). Under the assumptions of Theorem \ref{thm:closed-form}, let a safety-critical control algorithm induced by a CBF $h(x)$ act on the system. If $\mathbb{I}_{\Psi<0}[H,h]_{g g^{\top}}\leq0$ in a neighborhood of the equilibrium, the stability of the original system is preserved, and the closed-loop system converges to the largest invariant subset in $\{x \in \mathcal{D} | P_{h,\alpha}- [H,H]_{R}=0 \}$.
\end{proposition}
\begin{proof}
The power balance for the safety-critically controlled system is given by $\dot{H}=-[H,H]_R+P_{h,\alpha}$. Using $H$ as a candidate Lyapunov function, a sufficient condition for local stability preservation of the original system is that in a neighborhood of the equilibrium $P_{h,\alpha}\leq 0$, which is equivalent to $\mathbb{I}_{\Psi<0}[H,h]_{g g^{\top}}\leq0$. Due to time-invariance of closed-loop system, LaSalle's invariance principle can be invoked to conclude that the system will converge to the largest invariant subset in $\{x \in \mathcal{D} | P_{h,\alpha}- [H,H]_{R}=0 \}$.
\end{proof}

}
The condition $P_{h,\alpha}\leq 0$ forces the safety-critical controller to ``act as a damper", aiding dissipative dynamics in (\ref{eq:pH}) in the stabilization task. \markup{This control strategy is normally referred to as \textit{damping injection} in the pH framework. The closed-form expression (\ref{eq:safety-criticalPower}) can be used to design and interpret the non-trivial damping effects that specific CBFs have on the controlled system, similarly to the analysis in \cite{10136379}, where passivity, rather than stability, has been addressed.}

\subsubsection{\markup{CBF-based energy-aware designs}}
The power term induced by CBFs (\ref{eq:safety-criticalPower}) induces a constructive way to ``put energy back in control", in the sense envisioned in energy-based methods (see e.g., \cite{Ortega2001PuttingControl} and \cite{Stramigioli2015Energy-AwareRobotics}), going beyond stabilization purposes and possibly injecting positive power in the system if the task requires it. \markup{Since energy is a function of the state of the system, CBFs offer a constructive tool for designing tasks that involve regulating the closed-loop energy, such as those required for achieving periodic locomotion, which cannot be achieved through purely passive designs.}
The following motivating example shows the concept, and highlights the role of the $\mathcal{K}$ function $\alpha$ in the algorithm.

\begin{example}[\markup{Lower/Upper bound on total energy}]
\label{ex}
Consider the case $R=0$, that is, the system (\ref{eq:pH}) evolves on energetic isolines. 
Using a CBF $h(x)=c-H(x)$, i.e., $\mathcal{C}=\{x \in \mathcal{D} | H(x)\leq c\}$, the power equality (\ref{eq: dissipationDiscussion}) becomes:
\begin{equation}
\label{eq:limitEnergyCBF}
    \dot{H}=P_{h,\alpha}= - \mathbb{I}_{\Psi<0} (\alpha (c-H(x)), \qquad \Psi=  \alpha (c-H(x)).
\end{equation}
If $x \in \mathcal{C}$ then $P_{h,\alpha}=0$, while if $x \notin \mathcal{C}$ the safety-critical control extracts the negative power $P_{h}=\alpha( h)$, which can be modulated with $\alpha$. In case the $\alpha$ is linear, i.e., $\alpha(h)=\gamma h, \, \gamma \in \mathbb{R}^{+}$, the constant $\gamma$ acts as a proportional gain on the ``energy error" $c-H(x)$ measuring the distance from the closest energy value that the system is ultmately allowed to stay in. The (simple) stability of the minimum of the Hamiltonian is clearly preserved as $[h,H]_{gg^{\top}}=-[H,H]_{gg^{\top}}\leq 0$.
Using instead $h(x)=-c+H(x)$, i.e., $\mathcal{C}=\{x \in \mathcal{D} | H(x)\geq c\}$, the power equality (\ref{eq: dissipationDiscussion}) becomes:
\begin{equation}
\label{eq:limitEnergyCBF2}
    \dot{H}=P_{h,\alpha}= \mathbb{I}_{\Psi<0} (\alpha (c-H(x)), \qquad \Psi=  \alpha (H(x)-c).
\end{equation}
If $x \in \mathcal{C}$ then $P_{h,\alpha}=0$. If instead $x \notin \mathcal{C}$, $P_{h,\alpha}$ is positive, i.e., the safety-critical controller injects power (modulated by $\alpha$) into the system to reach $\mathcal{C}$. \markup{The system is shaped into a nonlinear oscillator at the desired energy level, a behavior achieved by a non passive control action.}
\end{example}

\section{Mechanical systems and energy-based CBFs}
\label{sec:4}

In order to better comprehend the proposed methodology and discuss how it complements standard energy-based approaches, let us specialise system (\ref{eq:pH}) to fully actuated mechanical systems, where we introduce the state $x=(q^{\top},p^{\top})^{\top} \in \mathbb{R}^{2n}$ as canonical Hamiltonian coordinates on the cotangent bundle of the configuration manifold of the system. 
In particular $q\in \mathbb{R}^{n}$ is the vector of generalized coordinates, $p\in \mathbb{R}^n$ denotes the generalized momenta, $p:=M(q)\dot q$, where $M(q)=M(q)^\top>0$ is the positive definite inertia matrix of the system. The equations of motions are given by \eqref{eq:pH} with 
$$
    J(x)-R(x) = \begin{bmatrix}
            0 & I_n\\
            -I_n & -D
        \end{bmatrix}, ~g(x) =\begin{bmatrix}
            0\\
            B
        \end{bmatrix}
$$
resulting in 
\begin{equation}\label{eq:mech_ph}
\begin{cases}
    \begin{aligned}
        \begin{bmatrix}
            \dot q\\
            \dot p
        \end{bmatrix}&=
        \begin{bmatrix}
            0 & I_n\\
            -I_n & -D
        \end{bmatrix}
        \begin{bmatrix}
            \frac{\partial H}{\partial q}\\
            \frac{\partial H}{\partial p}
        \end{bmatrix} +
        \begin{bmatrix}
            0\\
            B
\        \end{bmatrix}u\\
        y &= \begin{bmatrix}
            0 & B^{\top}
        \end{bmatrix} \begin{bmatrix}
            \frac{\partial H}{\partial q}\\
            \frac{\partial H}{\partial p}
        \end{bmatrix} =B^{\top} \dot{q}
    \end{aligned}
    \end{cases}
\end{equation}
where $H:\mathbb{R}^{2n} \to \mathbb{R}$ is the total \markup{mechanical} energy
\[
    H(q,p)= K_e(q,p) + V(q),
\]
where $K_e(q,p)=\frac{1}{2}p^\top M^{-1}(q)p$ is the kinetic energy, $V:\mathbb{R}^{n} \to \mathbb{R}$ maps the position state to conservative potentials (gravity, elastic effects), $D=D^\top \geq 0$ takes into account dissipative and friction effects, $B\in \mathbb{R}^{n \times n}$ is the (full rank) input matrix, $I_n$ indicates the $n \times n$ identity matrix and non specified dimensions of matrices are clear from the context. 

A safety-critical control action with CBF $h(q,p)$ specialises then the total power balance (\ref{eq: dissipationDiscussion}) into
\begin{equation}
\label{powerMech}
\dot{H}=   -\dot{q}^{\top}D{\dot{q}}  + P_{h,\alpha},
\end{equation}
where
\begin{equation}
\label{eq:PowerMechCbf}
P_{h,\alpha}=- \mathbb{I}_{\Psi<0} \frac{\dot{q}^{\top}BB^{\top} \frac{\partial h}{\partial p}}{\frac{\partial^{\top} h}{\partial p}BB^{\top} \frac{\partial h}{\partial p}} \Psi,
\end{equation}
and
\begin{equation}
\label{eq:Mechconstr}
\Psi=\{h,H\} -\frac{\partial^{\top} h}{\partial p}D{\dot{q}}+\alpha(h(q,p)),
\end{equation}
where we denoted $\{\cdot,\cdot\}$ the standard \textit{Poisson bracket}, which is a particular case of $\{\cdot,\cdot \}_J$ defined in Sec. \ref{sec:background}A.1. with the standard symplectic matrix $J$ of unconstrained mechanical systems displayed in (\ref{eq:mech_ph}).

The previous expressions allow to design safety-critical controllers with a qualitiative and quantitative insight on the effect of a certain CBF $h(q,p)$ on the power balance of the controlled mechanical system. A class of CBFs which deserves some discussion is 

\begin{equation}
\label{eq:genGenenergyCBF}
    h(q,p)=\pm K_e(q,p)+\bar{h}(q)+c,
\end{equation}
where $\bar{h}$ is \markup{any differentiable} function of the configuration variable and $c$ is a real number.

The sufficient condition of Proposition \ref{prop:prop} on stability-preserving safety-critical control (i.e., the condition on which safety-critical control \markup{acts as a damper}) reduces to 
\begin{equation}
    \mathbb{I}_{\Psi <0} \bigg( \dot{q}^{\top}BB^{\top} \frac{\partial h}{\partial p} \bigg)\leq 0.
\end{equation}
Since $\partial K_e(q,p)/\partial p=\dot{q}$, a subclass of CBFs in the form (\ref{eq:genGenenergyCBF}) which satisfy this condition is:
\begin{equation}
\label{eq:GenenergyCBF}
    h(q,p)=-K_e(q,p)+\bar{h}(q)+c,
\end{equation}
and the amount of instantaneous dissipated power results in: 
\begin{equation}
P_{h,\alpha}=\mathbb{I}_{\Psi < 0} \Psi \leq 0.
\end{equation}

\begin{remark}
     CBFs (\ref{eq:GenenergyCBF}) encompass those referred as \textit{energy-based CBFs} in \cite{Singletary2021Safety-CriticalSystems}, defined as (\ref{eq:GenenergyCBF}) with $c=0$ and $\bar{h}=\alpha_E h_{k}(q)$ where $h_k(q)$ is a \markup{CBF encoding} purely kinematic constraints (e.g., obstacle avoidance). In \cite{10136379} the safety-critical control induced by CBFs in the form (\ref{eq:GenenergyCBF}) was proven to preserve passivity of nominal controllers.
\end{remark}

Finally, we state the following insightful result.
\begin{proposition}
\label{prop:prop2}
Given a mechanical pH system (\ref{eq:mech_ph}) undergoing a safety-critical control induced by a CBF in the form (\ref{eq:genGenenergyCBF}), the expression for $\Psi$ in (\ref{eq:Mechconstr}) specialises to     
\end{proposition}

\begin{equation}
\label{eq:constraint}
\Psi=\mp \dot{V}+\dot{\bar{h}} +\dot{q}^{\top}D{\dot{q}}+\alpha(h(q,p))
\end{equation}
\begin{proof}
We use properties of skew-symmetry and bilinearity of the Poisson bracket to calculate the expression $\{h,H\}$ for (\ref{eq:genGenenergyCBF}):
$\{\pm K_e+\bar{h}, K_e+V \}=\{ \pm K_e,V \}+\{-K_e,\bar{h} \}=\{ \pm K_e, V \mp \bar{h}\}$. Using the definition of the bracket and the fact that $V$ and $\bar{h}$ do not depend on $p$, one gets $\{ \pm K_e, V \mp \bar{h}\}=\mp \frac{\partial ^{\top} K_e}{\partial p} \frac{\partial (V \mp \bar{h})}{\partial q}=\mp \frac{\partial^{\top}(V \mp \bar{h})}{\partial q} \dot{q}=\dot{\bar{h}} \mp \dot{V}.$
\end{proof}
\markup{Expression (\ref{eq:constraint})} provides interpretability in the structure of $\Psi$ \markup{and its dependence on the CBF parameter $\alpha$}. \markup{A negative $\Psi$ indicates that the safety-critical controller is active, and (\ref{eq:constraint}) is the injected or extracted power (depending on the sign of (\ref{eq:genGenenergyCBF})).} 




\section{Simulations \markup{and Discussion}}
\label{sec:sims}

\begin{figure}[t]
        \centering
        \begin{subfigure}[b]{.52\linewidth}
            \includegraphics[width=\linewidth]{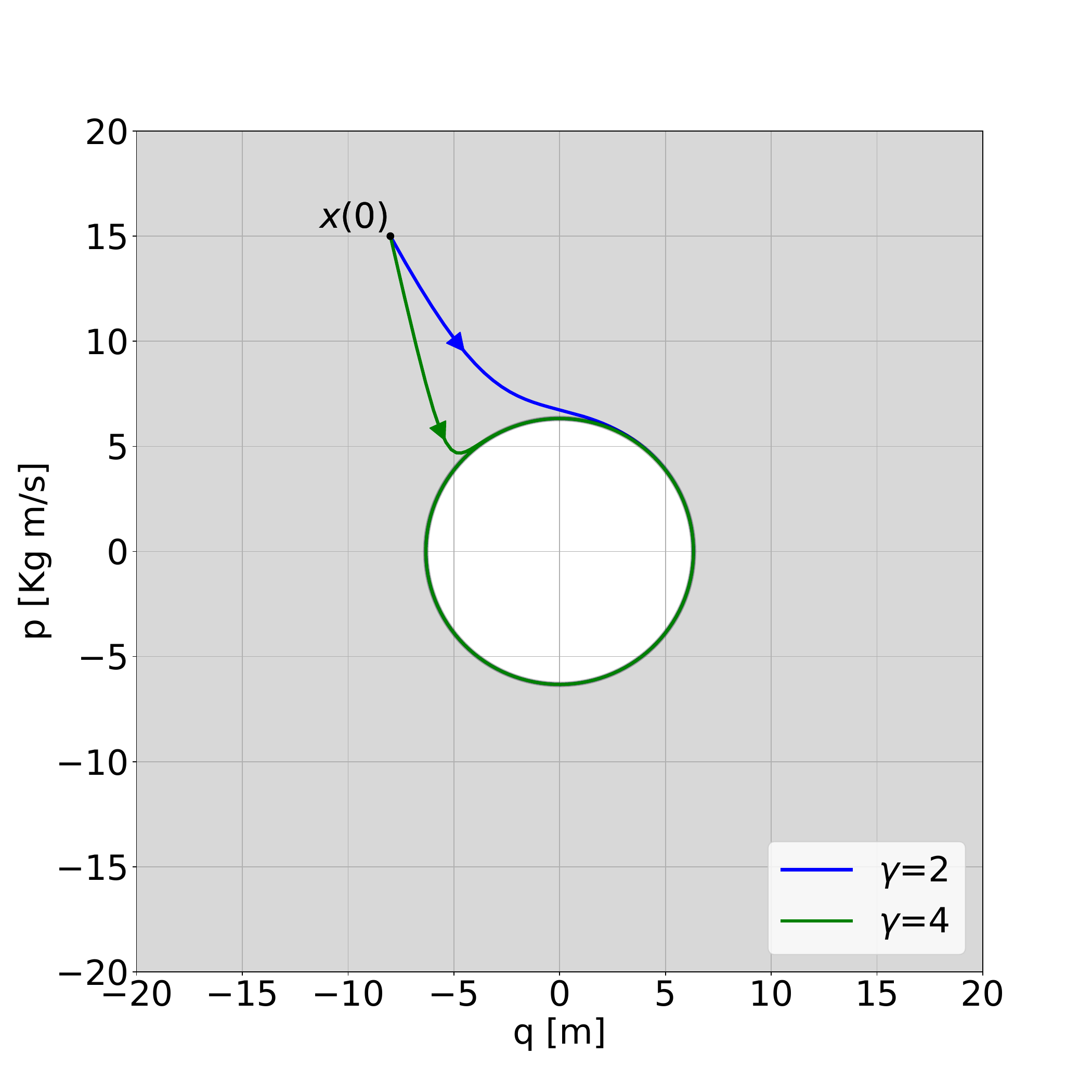}
            \label{fig:gull}
        \end{subfigure}\hspace{-5mm}
        \begin{subfigure}[b]{.52\linewidth}
            \includegraphics[width=\linewidth]{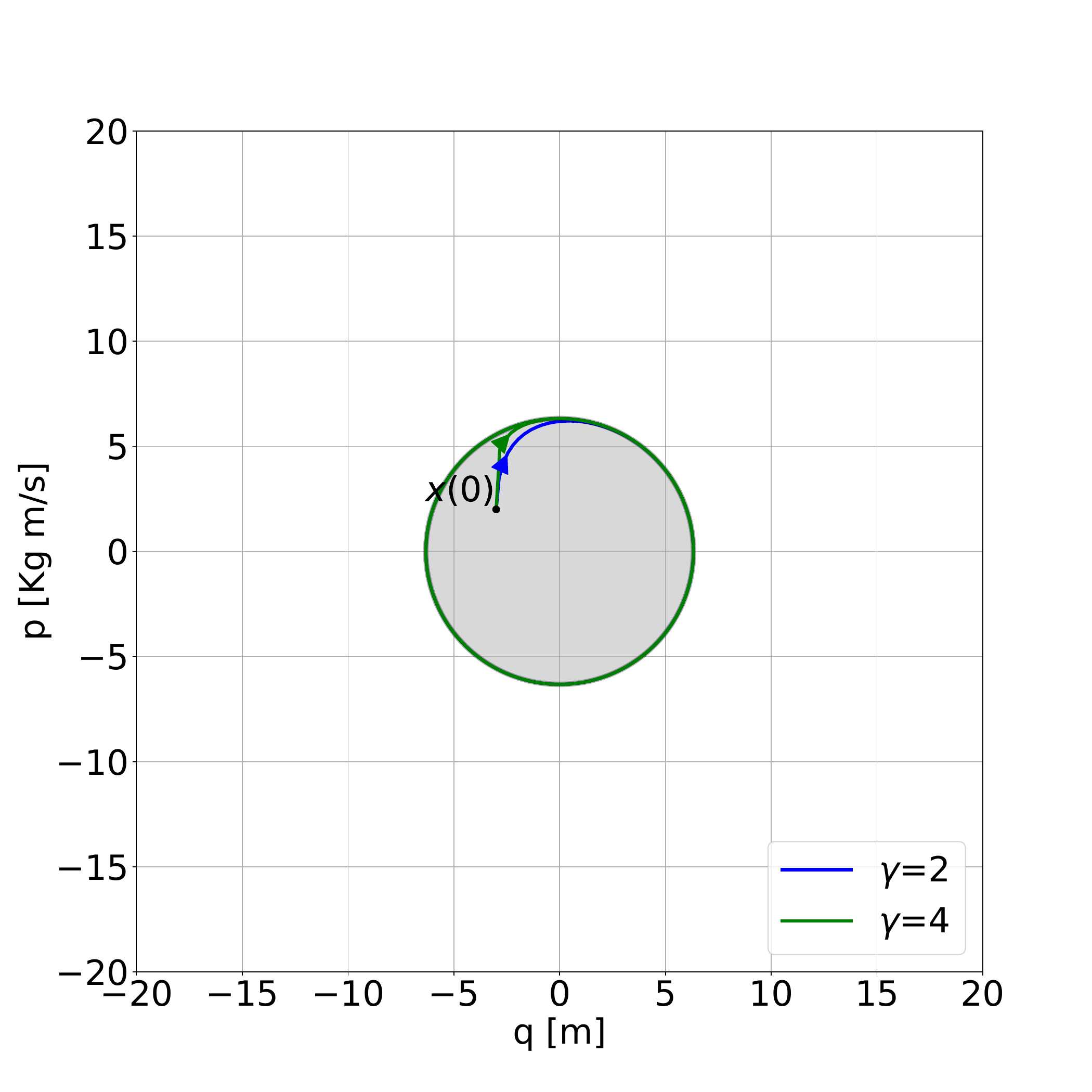}
            \label{fig:tiger}
        \end{subfigure}\vspace{-6mm}
        
        \begin{subfigure}[b]{.52\linewidth}
            \includegraphics[width=\linewidth]{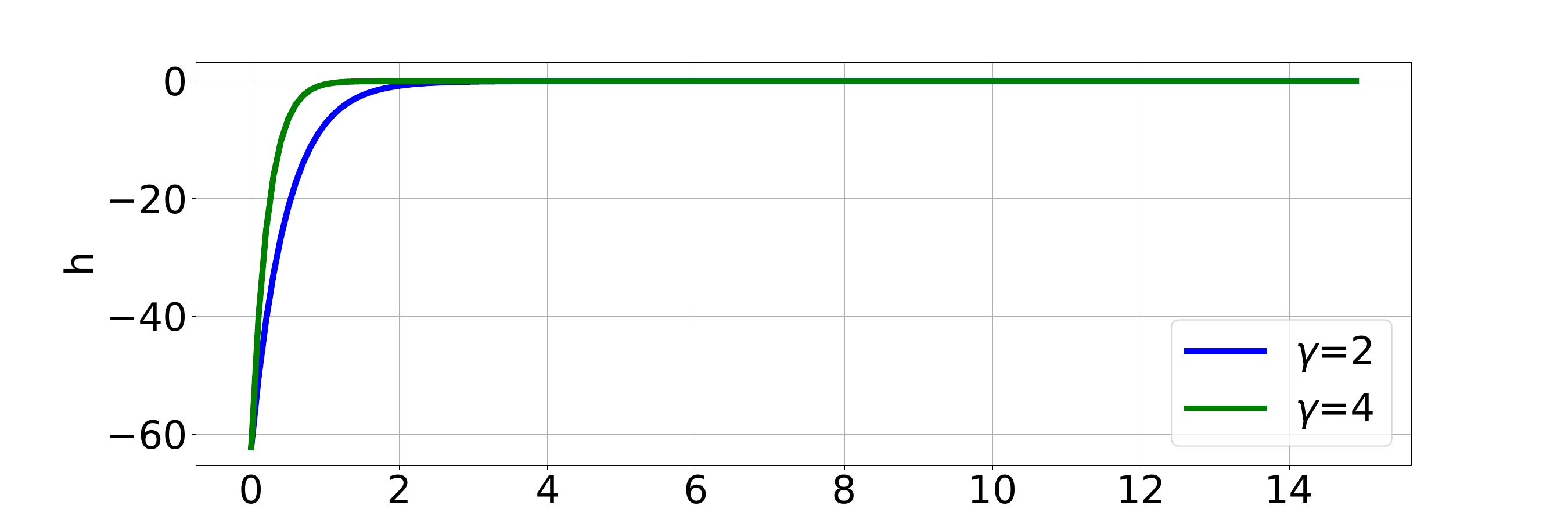}
            \label{fig:mouse} 
        \end{subfigure}\hspace{-5mm}
        \begin{subfigure}[b]{.52\linewidth}
            \includegraphics[width=\linewidth]{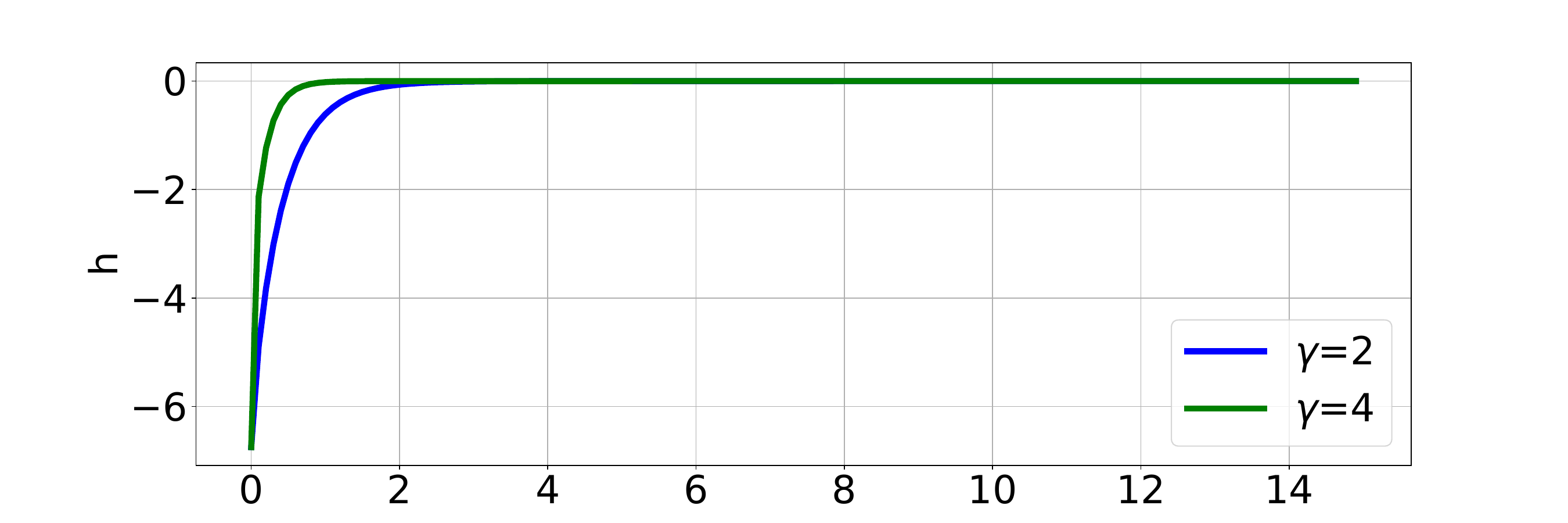}
            \label{fig:mouse}
        \end{subfigure}\vspace{-4mm}
        
        \begin{subfigure}[b]{.52\linewidth}
            \includegraphics[width=\linewidth]{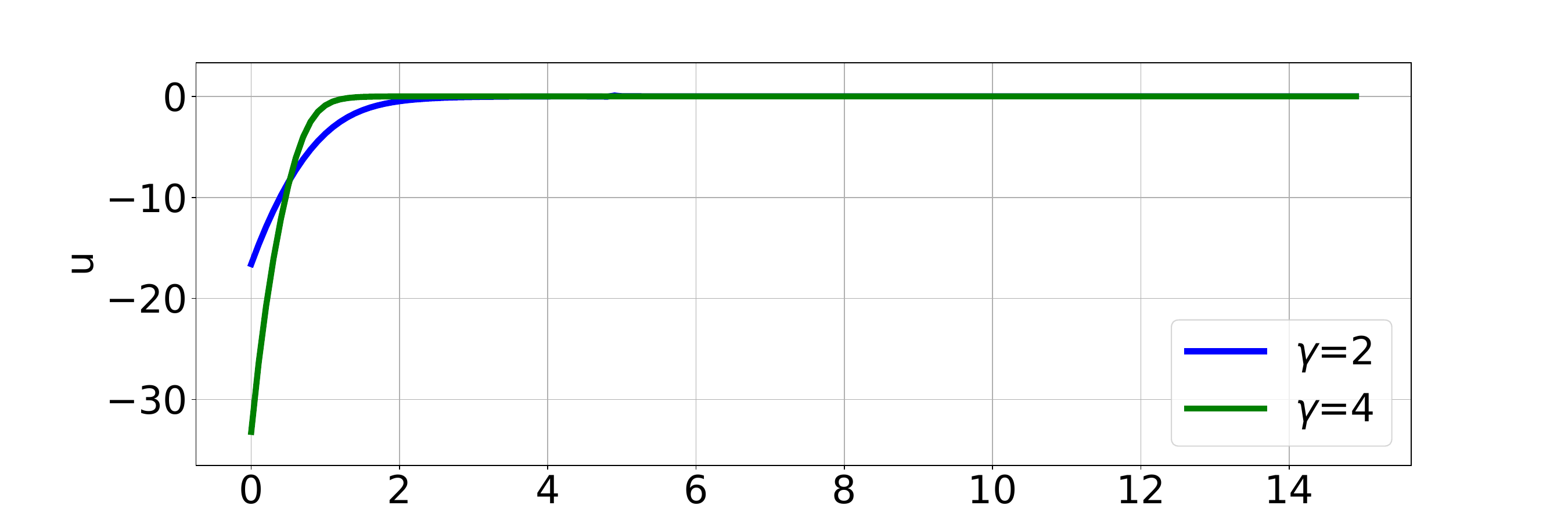}
            \label{fig:gull}
        \end{subfigure}\hspace{-5mm}
        \begin{subfigure}[b]{.52\linewidth}
            \includegraphics[width=\linewidth]{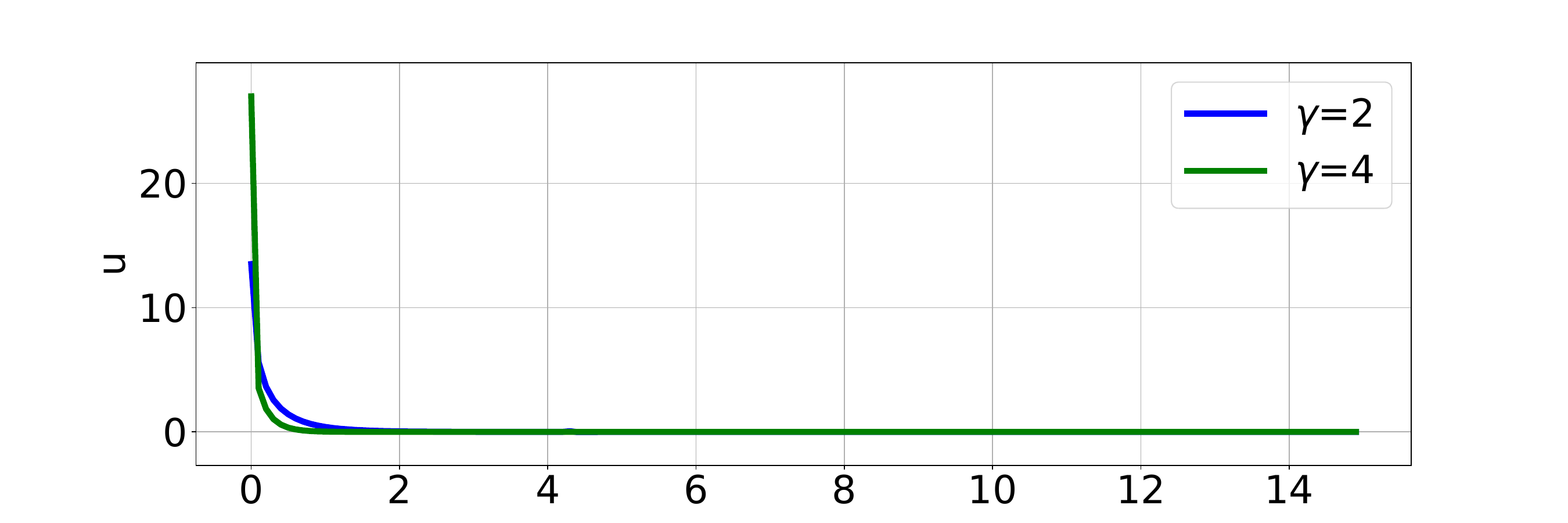}
            \label{fig:tiger}
        \end{subfigure}\vspace{-5mm}
        
        \begin{subfigure}[b]{.52\linewidth}
            \includegraphics[width=\linewidth]{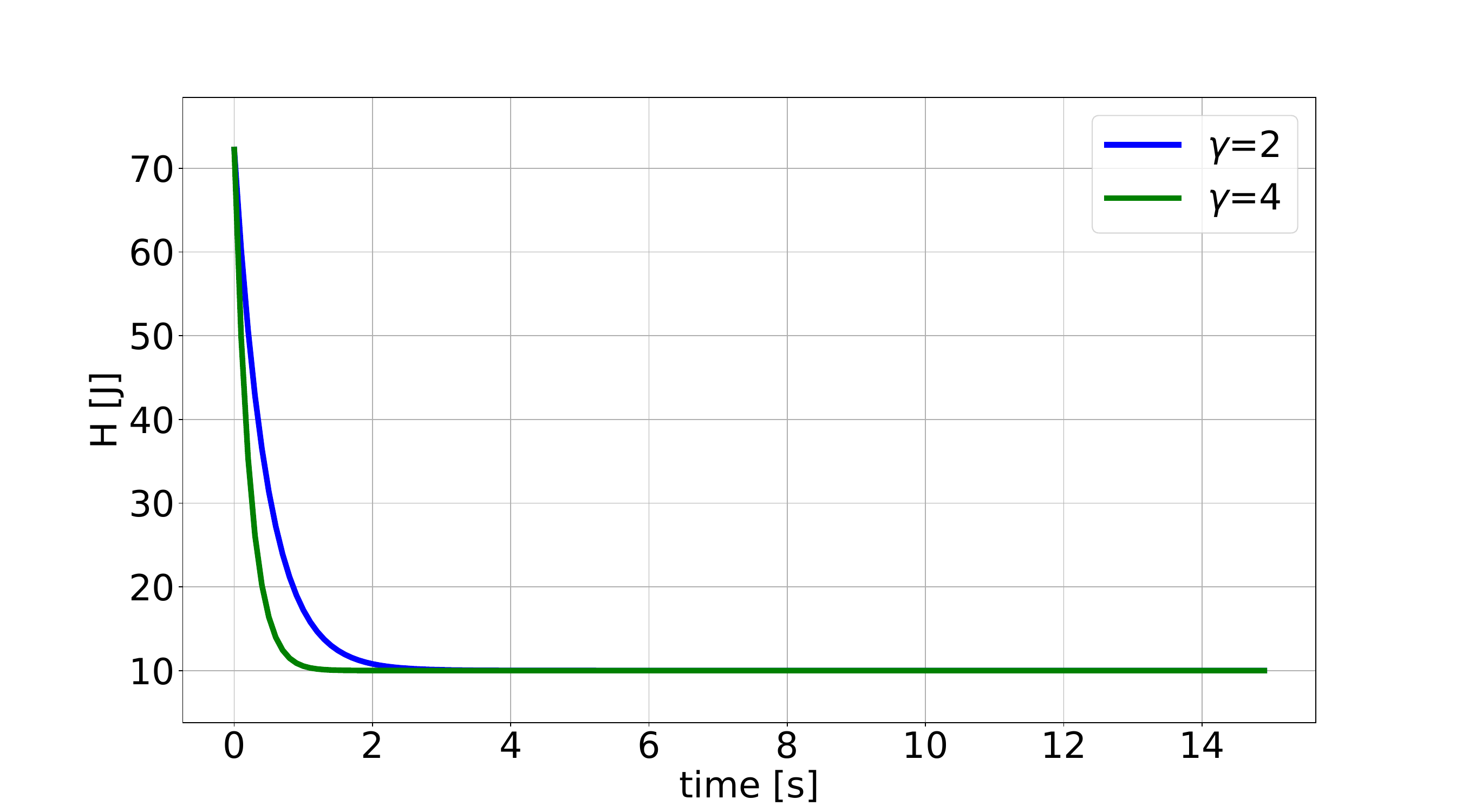}
            \label{fig:gull}
        \end{subfigure}\hspace{-5mm}
        \begin{subfigure}[b]{.52\linewidth}
            \includegraphics[width=\linewidth]{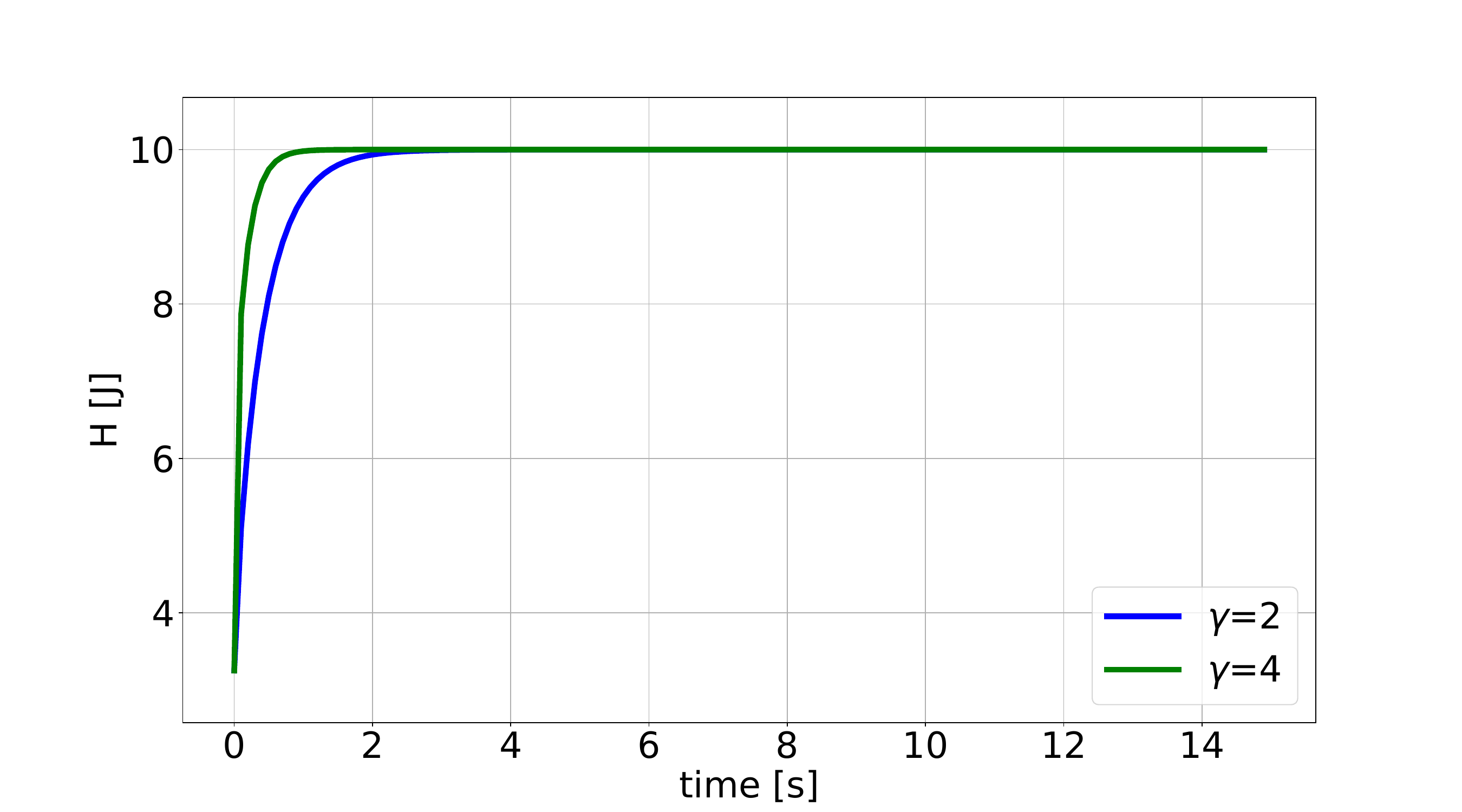}
            \label{fig:tiger}
        \end{subfigure}\vspace{-4mm}
        \caption{Mass-Spring system. Bounding energy from above (left) and below (right). From above: i) phase space trajectory, ii) CBF, iii) control input, and iv) total energy.} \label{fig:mass_spring_total}
        \vspace{-3mm}
\end{figure}

\begin{figure}[t]
\centering
\begin{subfigure}[b]{.52\linewidth}
\includegraphics[width=\linewidth]{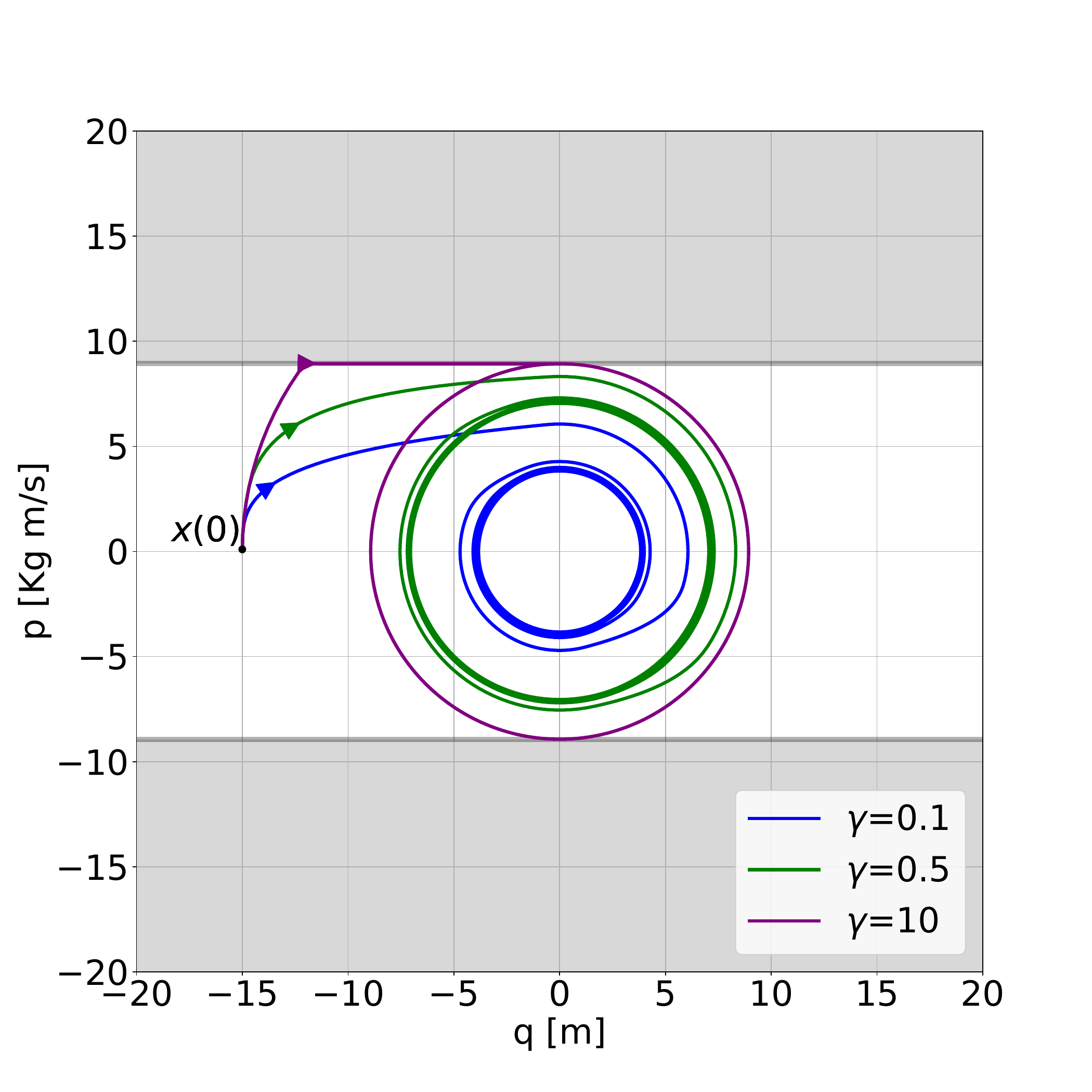}
\label{fig:gull}
\end{subfigure}\hspace{-5mm}
\begin{subfigure}[b]{.52\linewidth}
\includegraphics[width=\linewidth]{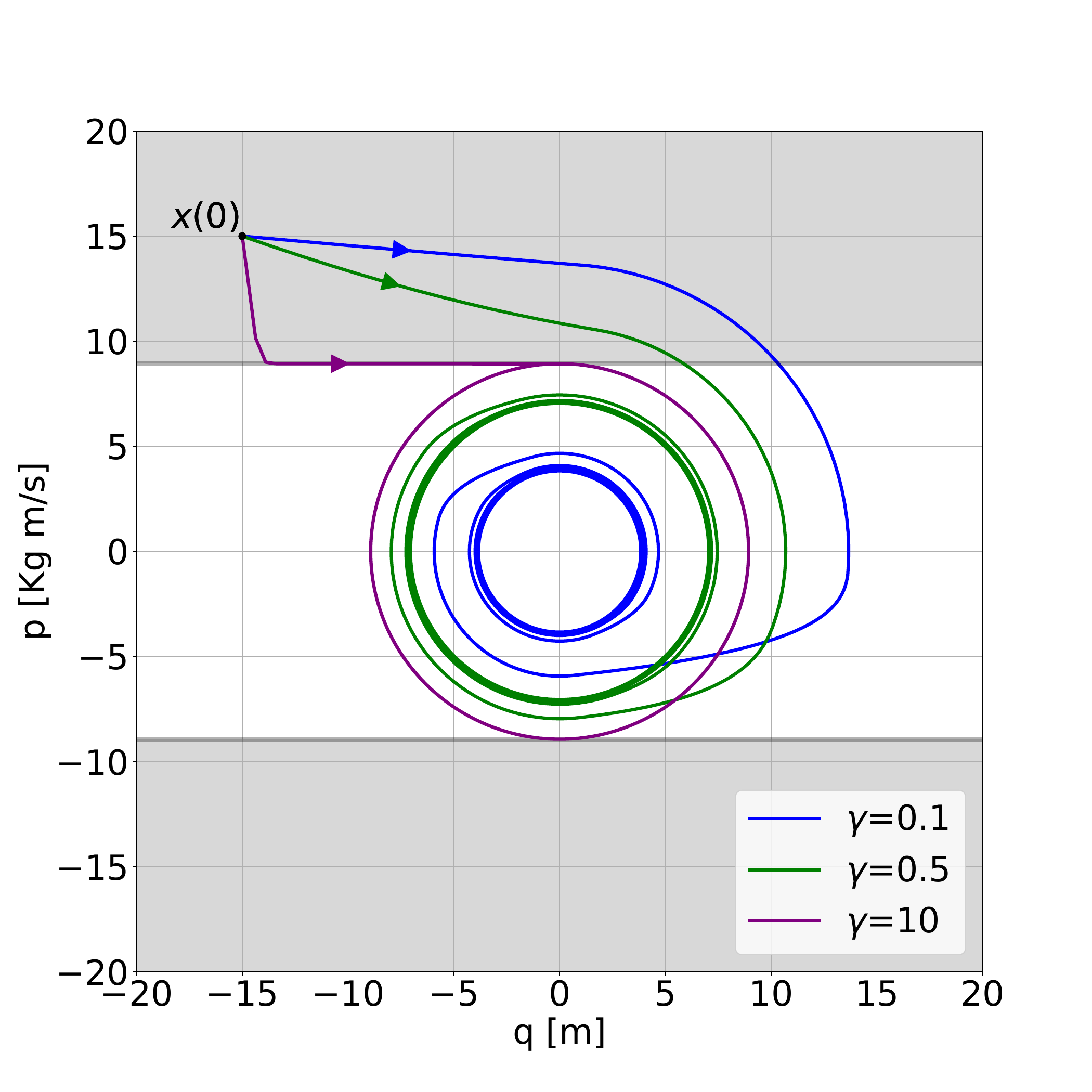}
\label{fig:tiger}
\end{subfigure}\vspace{-6mm}

\begin{subfigure}[b]{.52\linewidth}
\includegraphics[width=\linewidth]{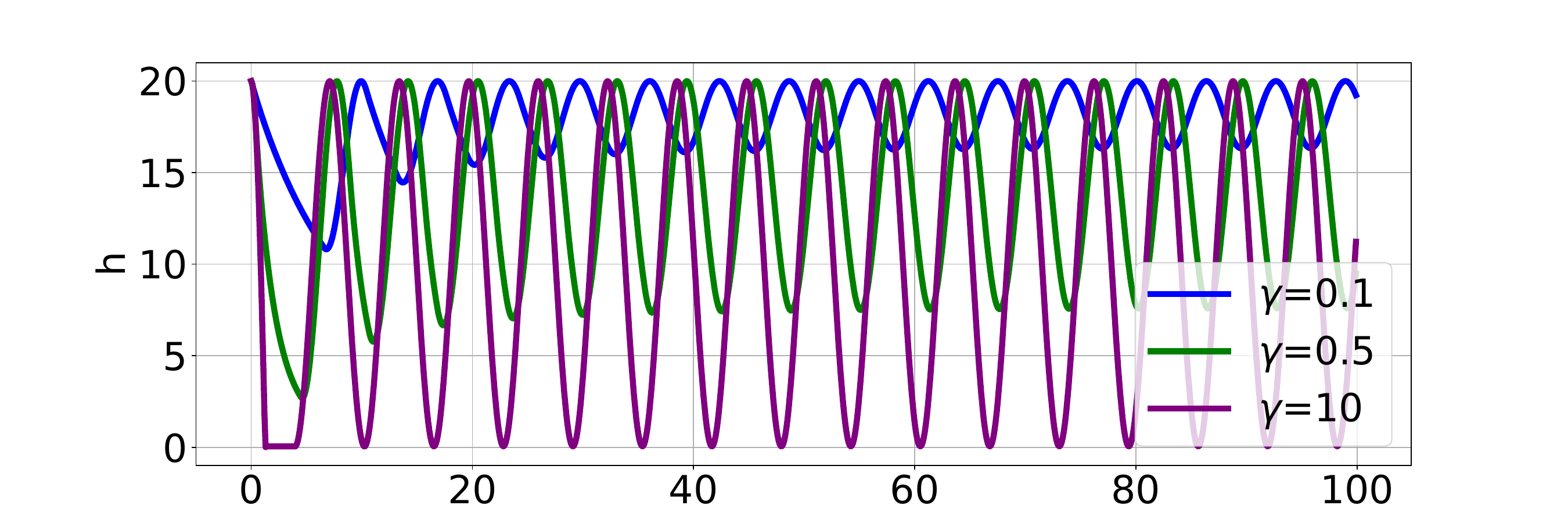}
\label{fig:mouse} 
\end{subfigure}\hspace{-5mm}
\begin{subfigure}[b]{.52\linewidth}
\includegraphics[width=\linewidth]{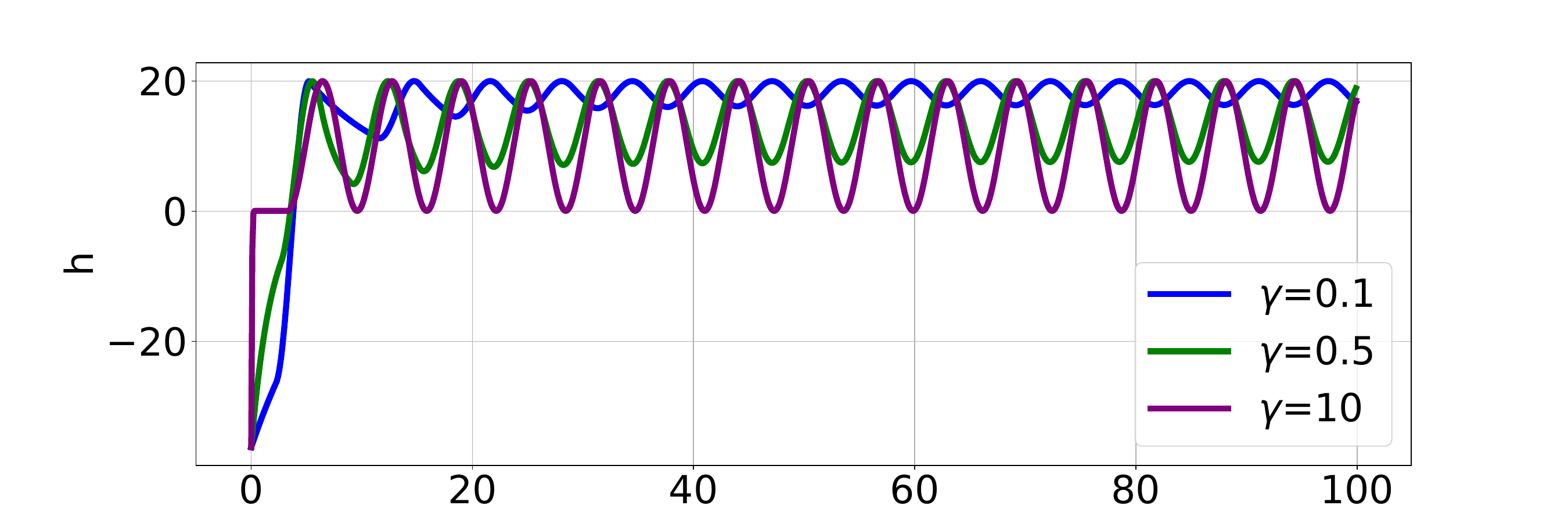}
\label{fig:mouse}
\end{subfigure}\vspace{-4mm}

\begin{subfigure}[b]{.52\linewidth}
\includegraphics[width=\linewidth]{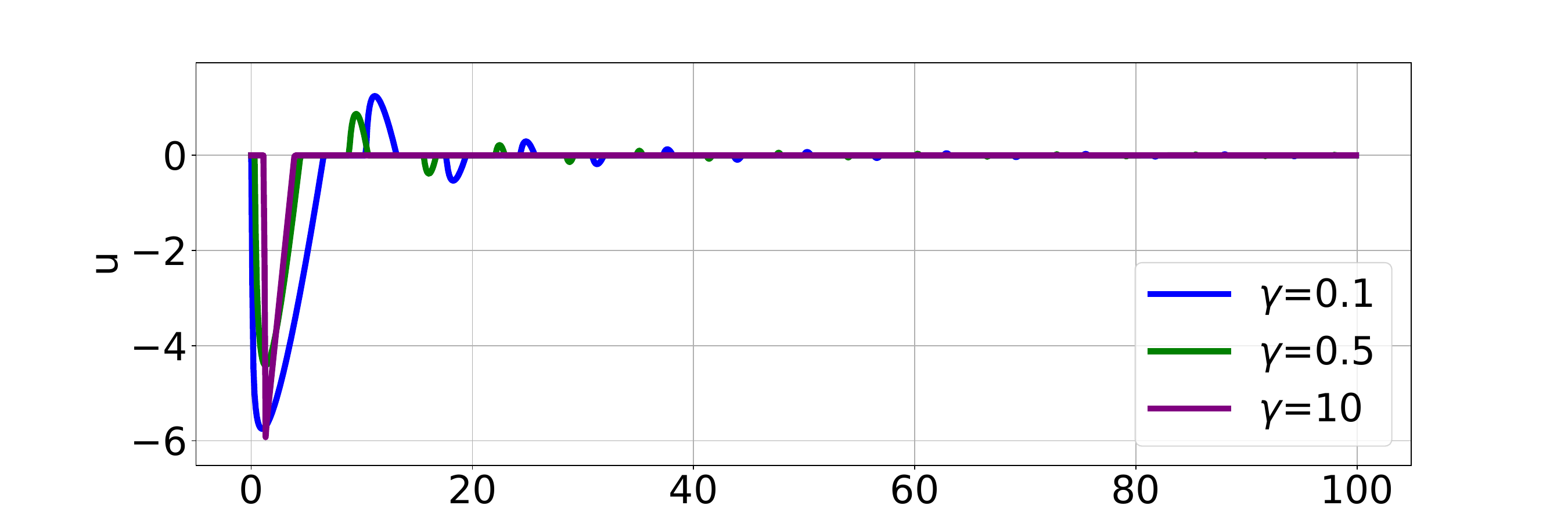}
\label{fig:gull}
\end{subfigure}\hspace{-5mm}
\begin{subfigure}[b]{.52\linewidth}
\includegraphics[width=\linewidth]{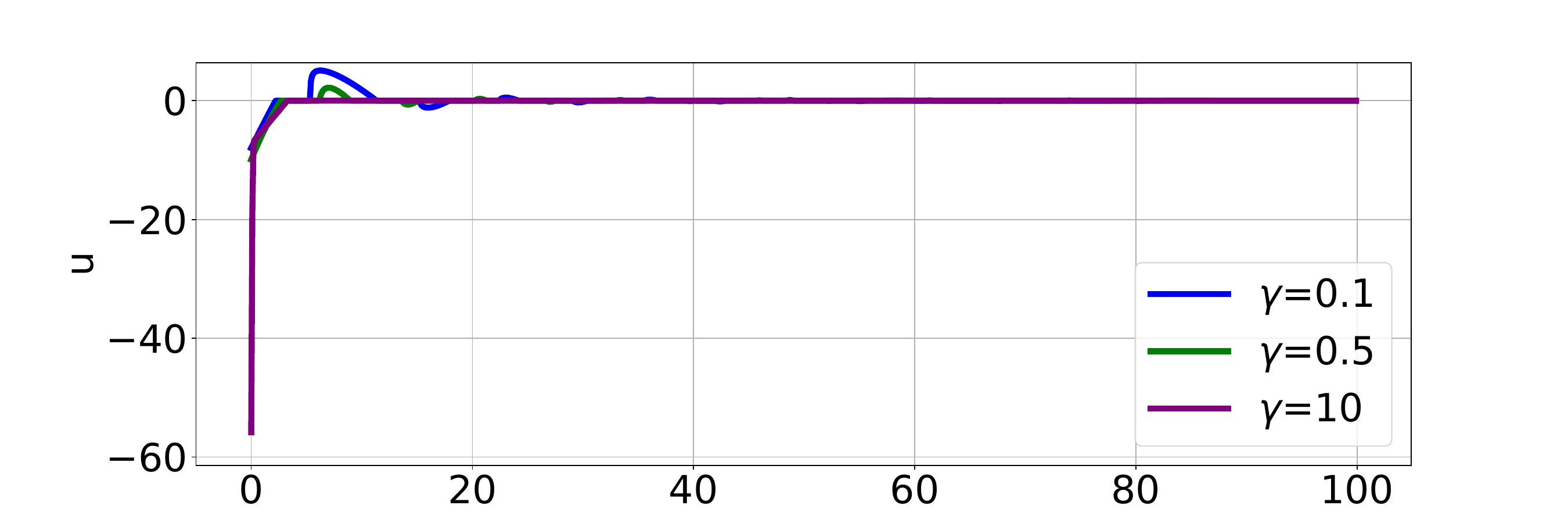}
\label{fig:tiger}
\end{subfigure}\vspace{-5mm}

\begin{subfigure}[b]{.52\linewidth}
\includegraphics[width=\linewidth]{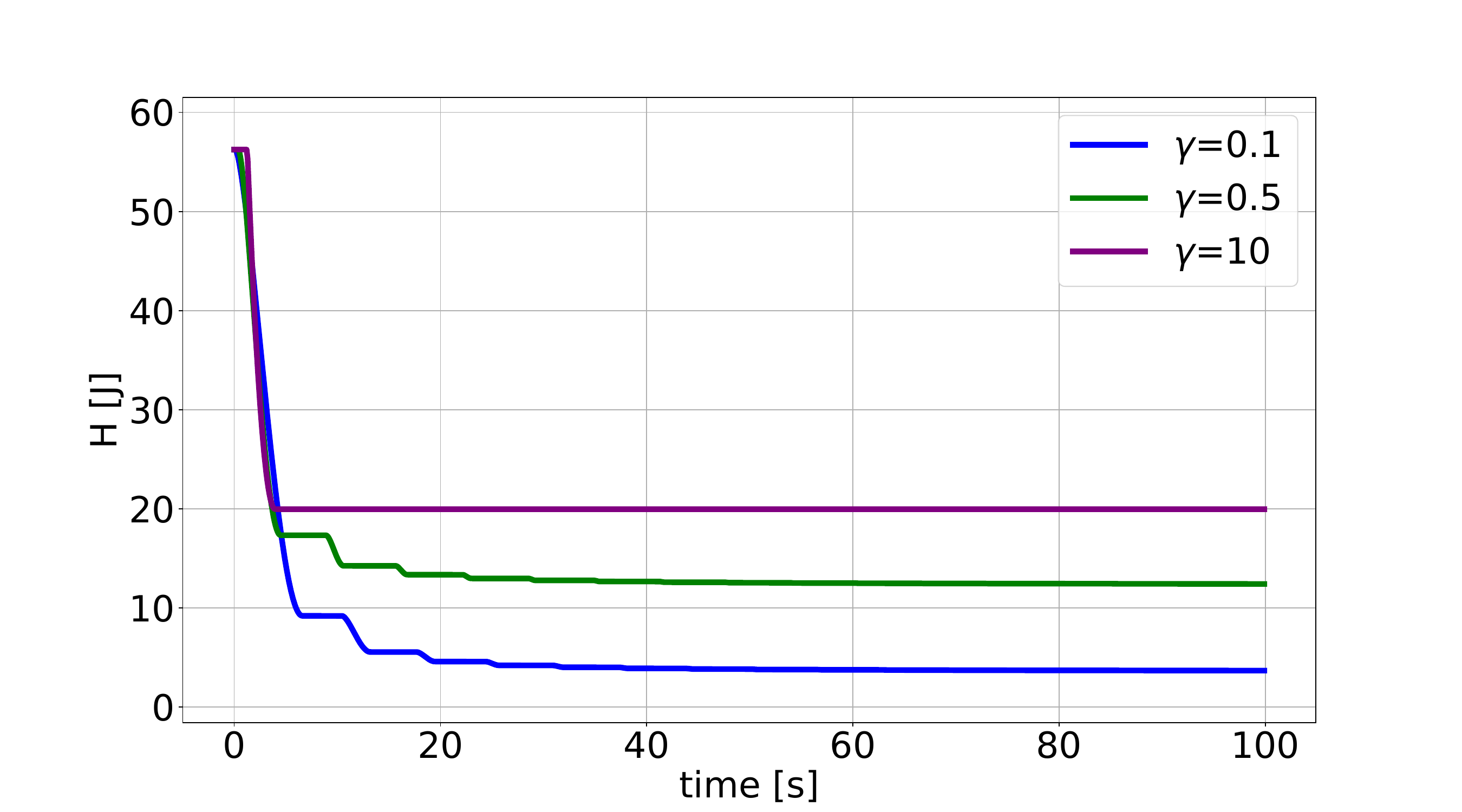}
\label{fig:gull}
\end{subfigure}\hspace{-5mm}
\begin{subfigure}[b]{.52\linewidth}
\includegraphics[width=\linewidth]{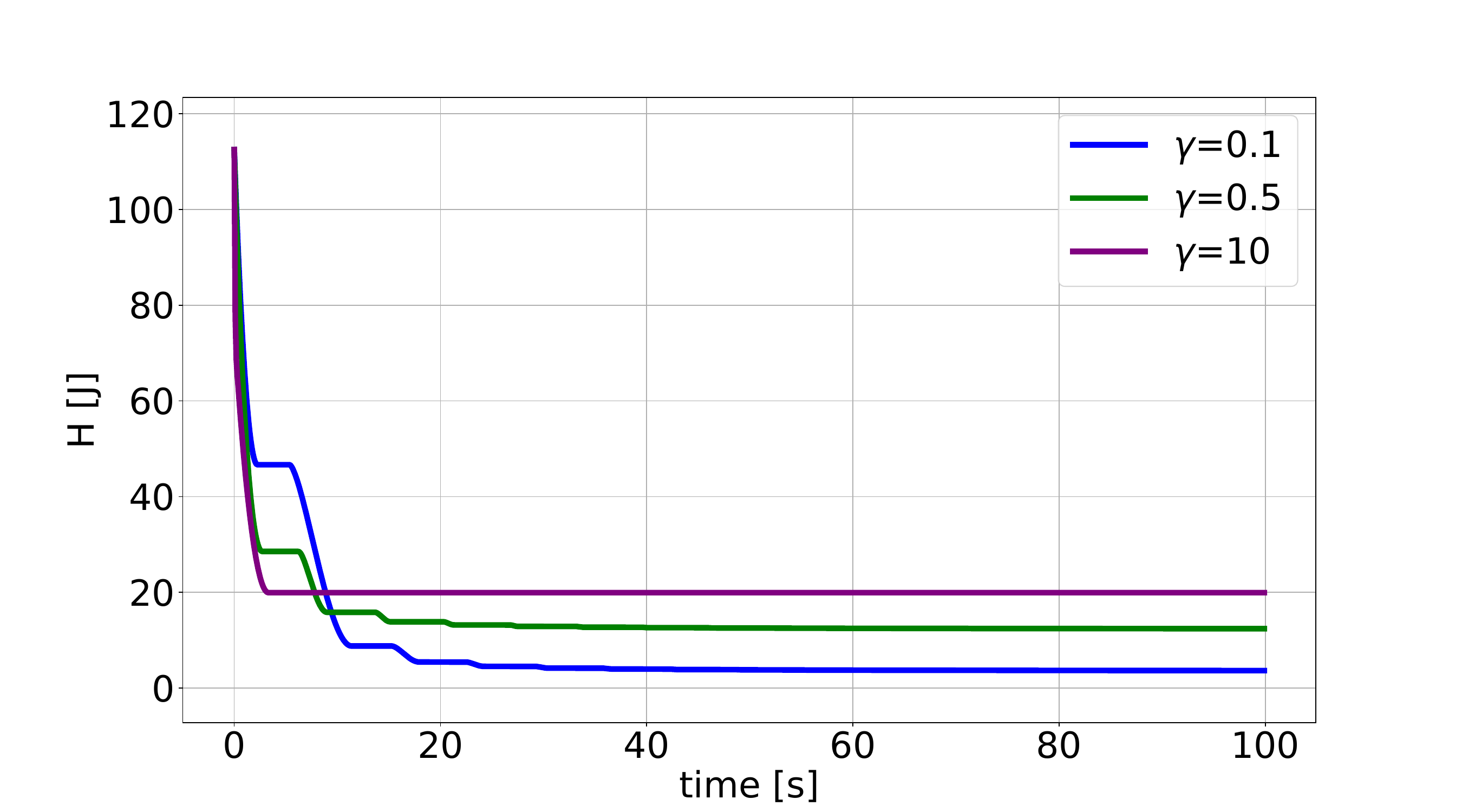}
\label{fig:tiger}
\end{subfigure}\vspace{-4mm}
\caption{Limiting kinetic energy from inside (left) and outside (right) the safe set. From above: i) phase space trajectory, ii) CBF, iii) control input, and iv) total energy.}
\label{fig:mass-spring_kinetic}
\end{figure}




In this section, we display and discuss numerical simulations of mechanical systems undergoing safety-criticl control with CBFs in the form (\ref{eq:genGenenergyCBF}). When the phase space is plotted, the safe set is represented in white and its complement is represented in gray.  The function $\alpha$ is chosen as $\alpha(h)=\gamma h$, and simulations are performed for different positive scalars $\gamma$. \markup{When not reported, units are clear from the context.}
\subsubsection{Bounding total energy from below and above}
We initially consider as a plant a simple oscillator (mass-spring) with no dissipation, i.e., system (\ref{eq:mech_ph}) with $D=0$, $B=1$,  $H=\frac{p^2}{2m}+\frac{k q^2}{2}$, $k=0.5$ and $m=2$. The results are shown in Fig. \ref{fig:mass_spring_total}. The left column implements a safety-critical controller given by a CBF bounding the total energy from above ($h=-H+c$) while the second column uses a CBF bounding the total energy from below ($h=H-c$), with energy limit $c=10$. The proposed approach generates new ways to implement mechanical oscillations at the desired energy level. As theoretically derived from example \ref{ex}, we see that the safety critical input is triggered once the system is outside the safe set, and we see how the $\gamma$ parameter can be set to regulate the power the controller dissipates or injects.
\subsubsection{Bounding kinetic energy}
Fig. \ref{fig:mass-spring_kinetic} shows simulations with the same plant, but with a CBF designed to limit only the kinetic energy, i.e., $h=-K_e+c$ where $c=20$. The left column corresponds to experiments in which the system starts in the safe set, and the right column to one in which the system starts outside the safe set. The results show both the invariance and the asymptotic stability of the safe set. Looking at the total energy, it appears clearly that, as predicted by the theory, the safety-critical controller acts purely as a damper, with the amount of damping which be regulated with $\gamma$. The activation of safety-critical effects is no longer trivially related to whether the system is in the safe set or not, and can be inspected by checking the positivity of the constraint (\ref{eq:constraint}).

 \subsubsection{\markup{Energy-based nonlinear oscillator}}
\markup{In Fig. \ref{fig:double_pendulum_pump} we show simulations for a nonlinear plant: a double pendulum with masses $m_1=m_2=1.5$ and lengths $l_1=l_2=1$, with natural dissipation $b=0.3$ in both the joints. This system is described by (\ref{eq:mech_ph}) with $D=b I_2$, $B=I_2$ and $H=\frac{p^{\top}M^{-1}(q)p}{2}+V(q)$, with $V(q)$ the gravitational potential, taking negative values below the first joint. The CBF is designed to limit the total energy, i.e., $h=-H+c$ where $c=-40$, and it is activated after $10s$ from the start of the simulation. The system, which in the first phase was dissipating its initial mechanical energy, converges to the safe set after the CBF activation. Interestingly, the controlled system converges to a limit cycle: the natural dissipation and the power injected by the CBF counterbalance each other on the time period on which the limit cycle emerges, resembling a periodic locomotion task. This procedure constitutes a new, energy-based, way to generate periodic motions in mechanical systems.}
\begin{figure}[t]
    \centering
    \begin{minipage}{.48\linewidth} 
        \raggedright
        \begin{subfigure}[b]{\linewidth}
            \includegraphics[width=1.1\linewidth]{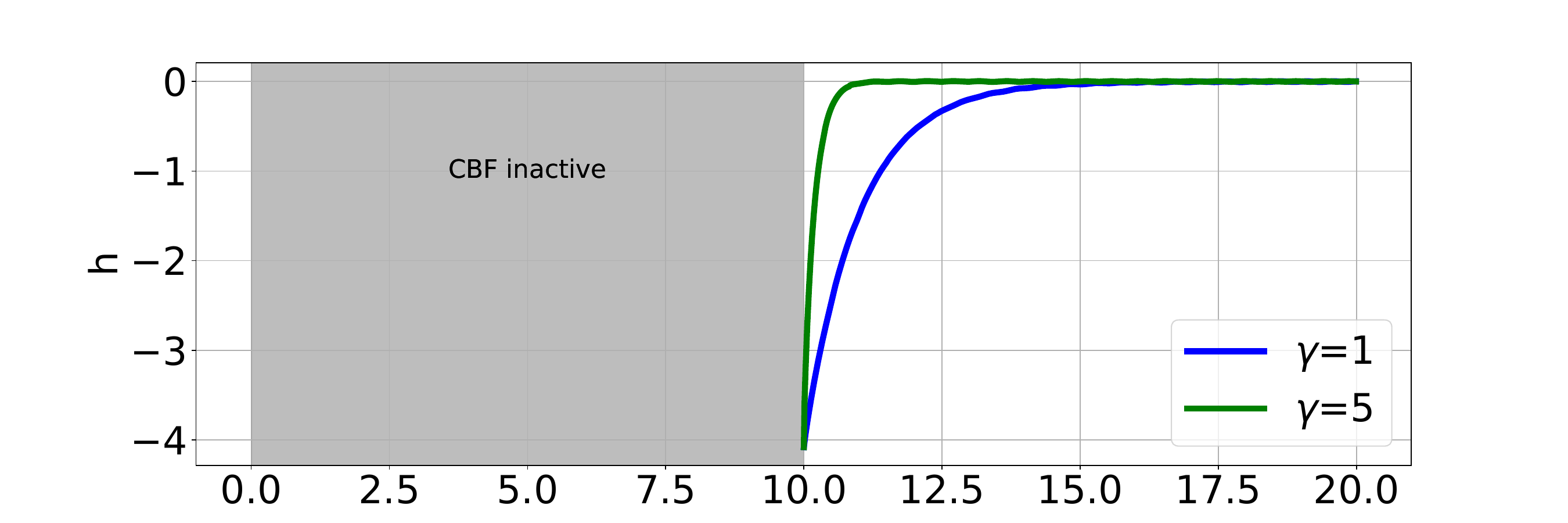}
            \label{fig:total_energy}
        \end{subfigure}\vspace{-4mm}
        
        \begin{subfigure}[b]{\linewidth}
            \includegraphics[width=1.1\linewidth]{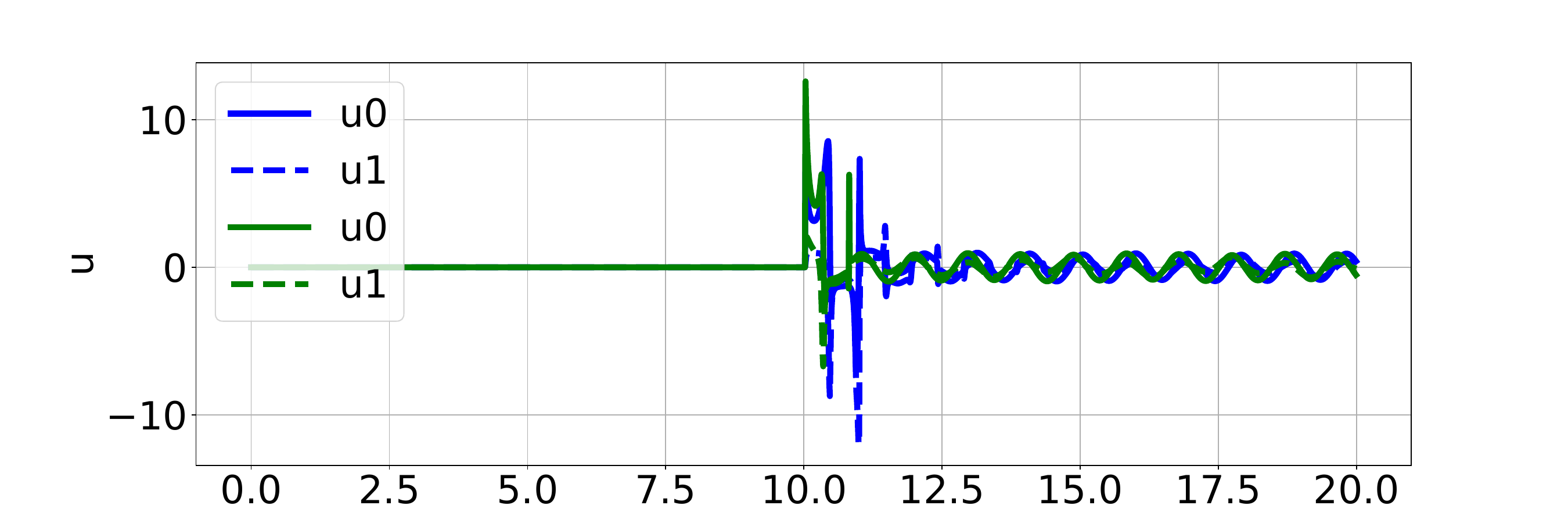}
            \label{fig:cbf} 
        \end{subfigure}\vspace{-3mm}
        
        \begin{subfigure}[b]{\linewidth}
            \includegraphics[width=1.1\linewidth]{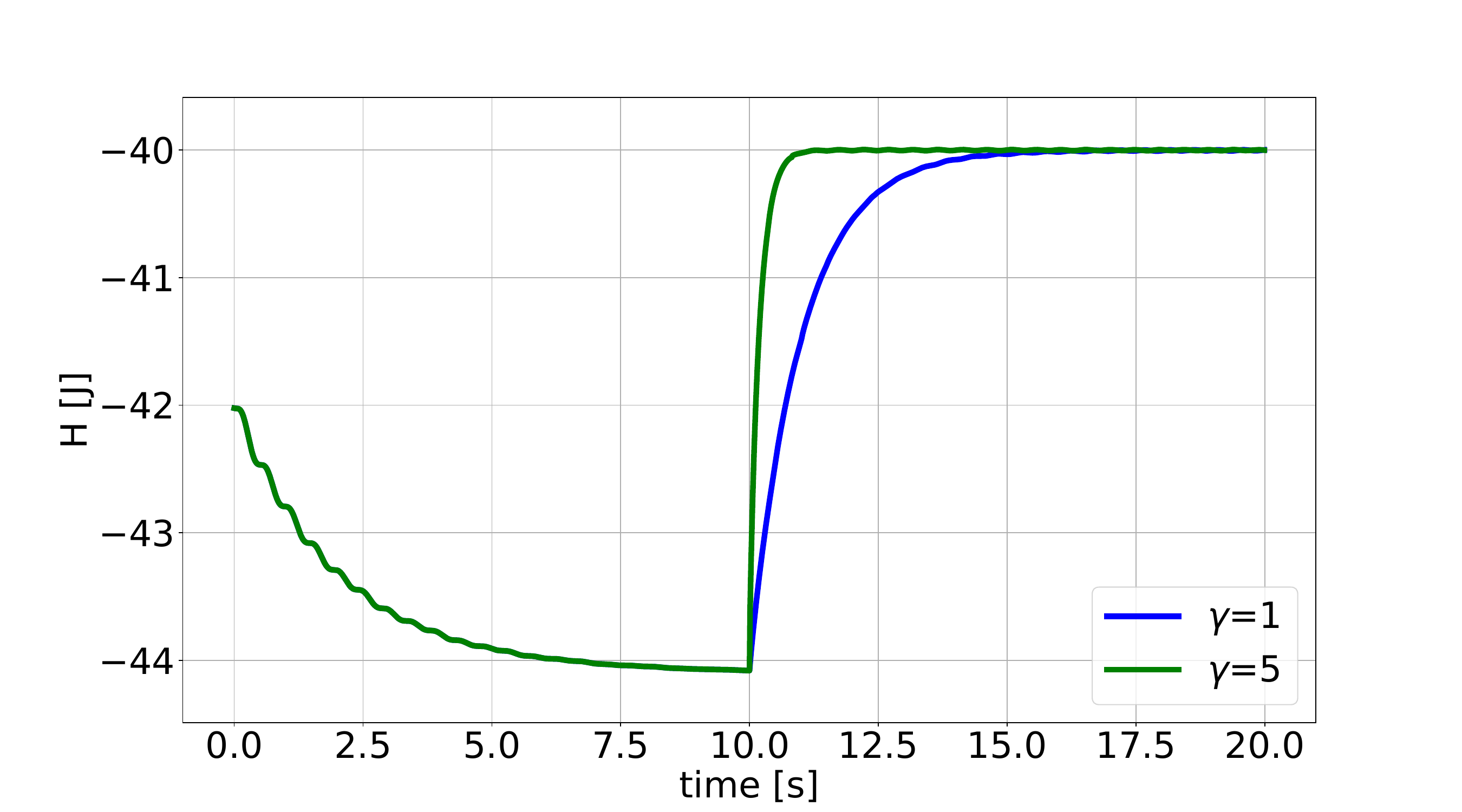}
            \label{fig:control}
        \end{subfigure} 
    \end{minipage}%
    \hfill
    \begin{minipage}{.52\linewidth} 
        \centering \vspace{-3mm}
        \begin{subfigure}[b]{\linewidth}
            \includegraphics[width=\linewidth]{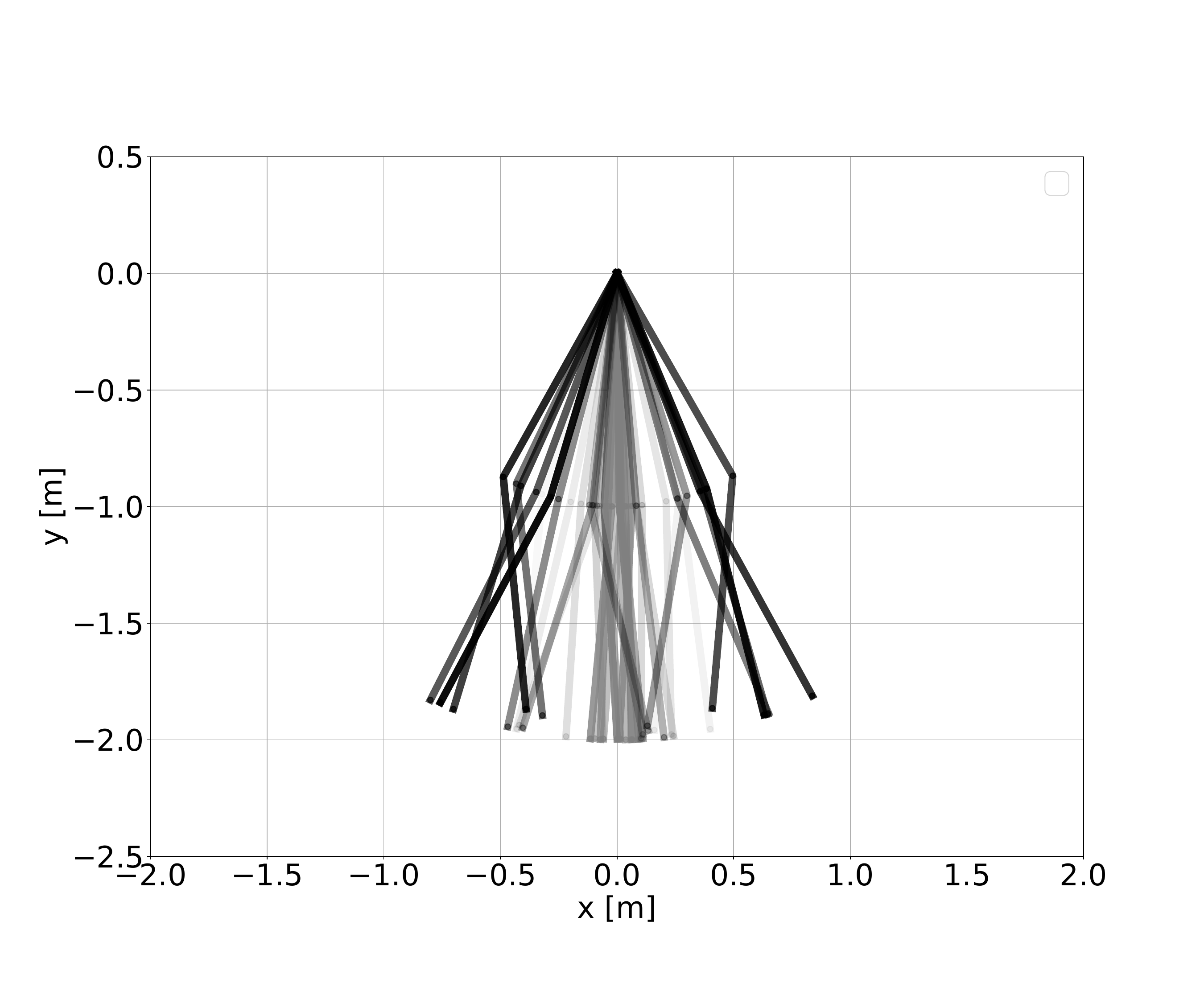}
            \label{fig:cartesian}
        \end{subfigure} \vspace{-10mm}
        
        \begin{subfigure}[b]{\linewidth}
            \includegraphics[width=\linewidth]{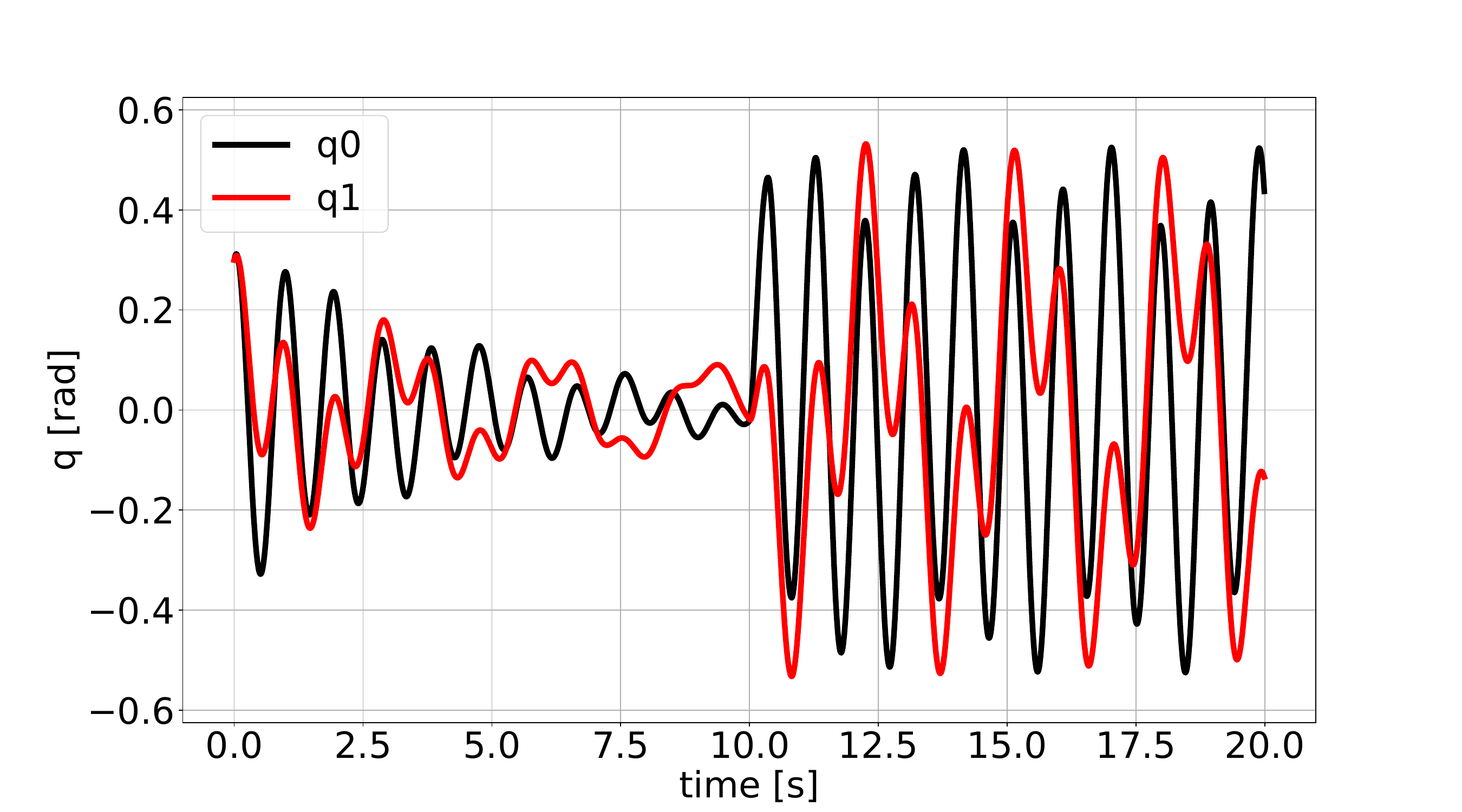}
            \label{fig:state}
        \end{subfigure}  
    \end{minipage}
    
    \vspace{-4mm}

        \begin{subfigure}[b]{.52\linewidth}
            \includegraphics[width=\linewidth]{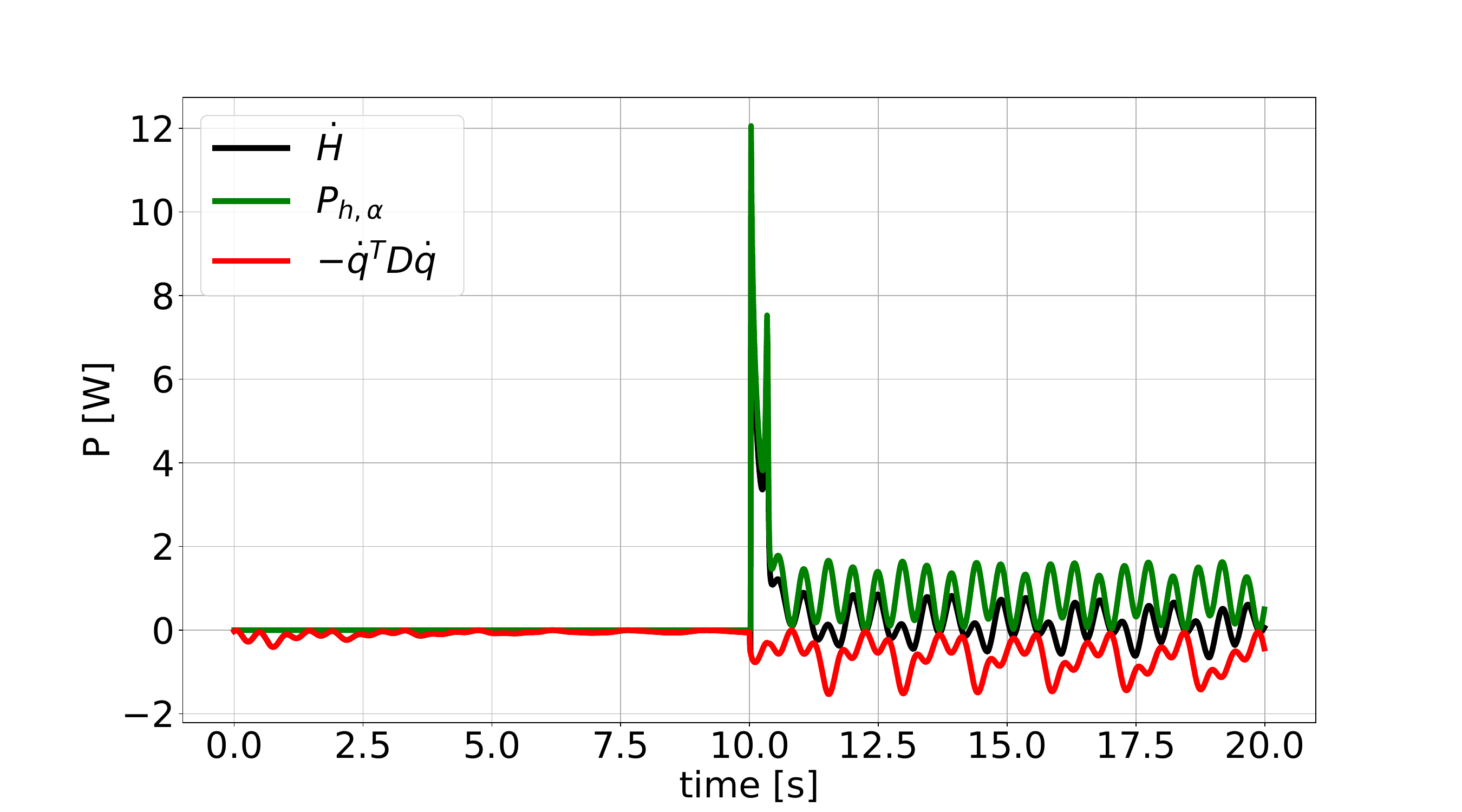}
            \label{fig:gull}
        \end{subfigure}\hspace{-5mm}
        \begin{subfigure}[b]{.52\linewidth}
            \includegraphics[width=\linewidth]{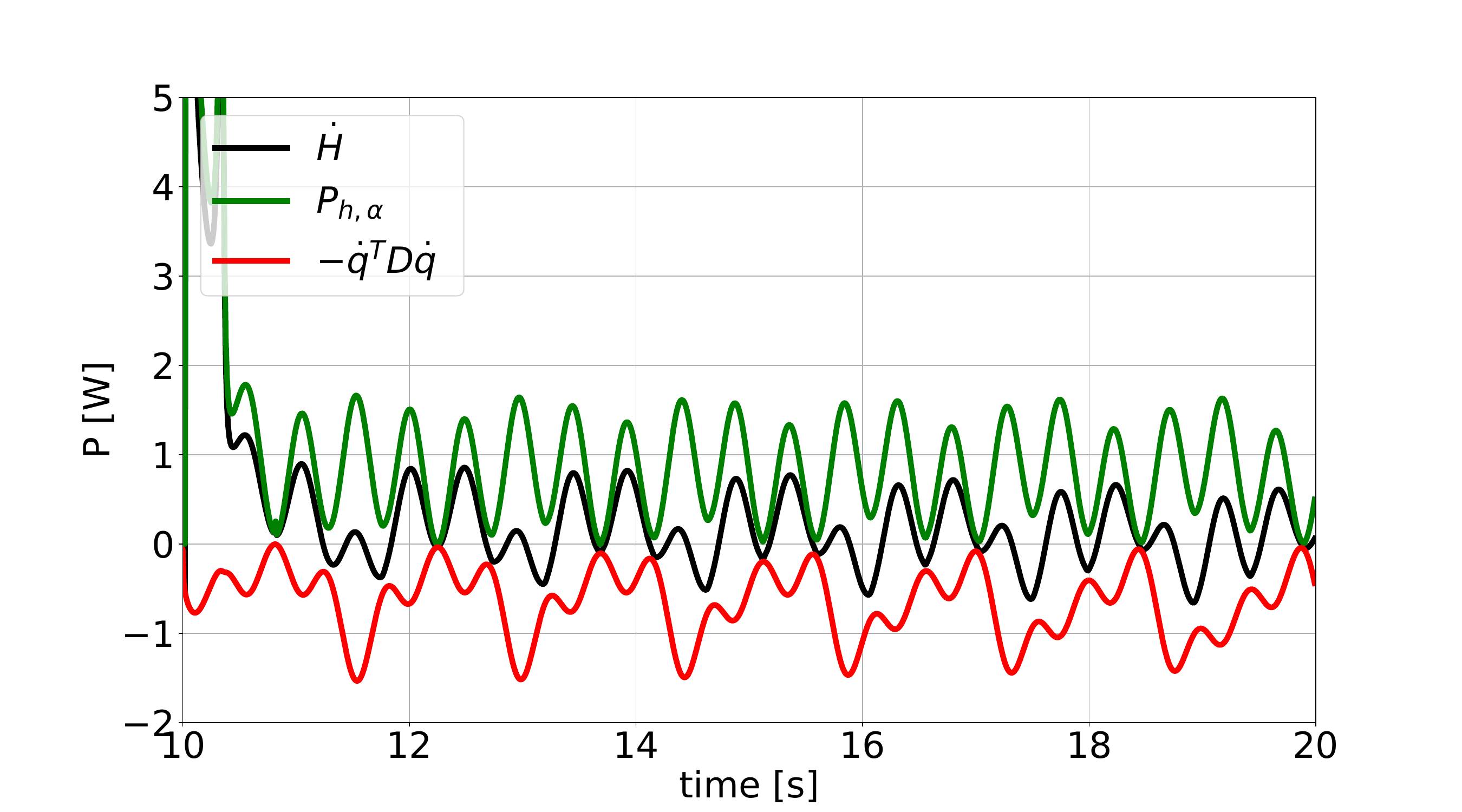}
            \label{fig:tiger}
        \end{subfigure}\vspace{-4mm}
        
        \caption{\markup{Double pendulum. CBF imposing total energy lower limit is activated after 10 \si{s}. From above, left column: i) CBF, ii) safety-critical control input, iii) total energy, and iv) power terms in (\ref{powerMech}), detailed on its right. Top right: cartesian snapshots of the system before (gray) and after (black) the CBF activation, and joint trajectory below.}} \label{fig:double_pendulum_pump}
        \vspace{-3mm}
\end{figure}

\section{Conclusions}
\label{sec:conc}
In this paper we gave constructive tools to study the qualitative and quantitative effect of safety-critical control schemes implemented with CBFs on the energy balance of controlled physical systems. The analysis led to novel energy-aware schemes, such as selective damping injection mechanisms and active control strategies that inject energy into the controlled system to achieve desired closed-loop behaviors. 

\renewcommand*{\bibfont}{\small}
\printbibliography

\end{document}